\documentclass[letterpaper,twocolumn,10pt]{article}

\usepackage[available,functional,reproduced]{usenixbadges}

\usepackage{tikz}
\usepackage{usenix-2020-09}

\usepackage{hyperref}
\usepackage{paralist}
\usepackage{titling}
\usepackage{url}

\usepackage{filecontents}

\usepackage{makeidx} 
\usepackage{multirow}
\usepackage{graphicx}
\usepackage{amsmath}
\usepackage{amssymb}
\usepackage{amsfonts}

\usepackage{amsthm}
\usepackage{mathtools}
\usepackage{url}
\usepackage{rotating}
\usepackage{enumitem}
\usepackage{lscape}
\usepackage{bm}
\usepackage{braket} 

\usepackage{algorithm}      
\usepackage{algpseudocode}  

\usepackage{makecell}

\usepackage{float}

\usepackage{booktabs}

\usepackage{listings}
\usepackage{color}
\allowdisplaybreaks[4] 
\usepackage{autobreak}

\usepackage{graphicx}
\usepackage{lineno}

\usepackage{pgfplots, pgfplotstable}
    \pgfplotsset{compat=1.17}

\newtheorem{lemma}{Lemma}

\newcommand{\GF}[1]{\mathbb{F}_{#1}}
\newcommand{\myset}[1]{{#1}}
\newcommand{\pr}[1]{\text{\upshape #1}}

\newcommand{\wt}[1]{w_H(#1)\xspace}

\renewcommand{\vec}[1]{\mathbf{#1}\xspace}
\newcommand{\mat}[1]{\mathbf{#1}\xspace}

\newcommand{\ldpc}{$\mathsf{LDPC\text{-}PRC}$\xspace}

\algnewcommand\algorithmicparam{\textbf{Parameter:}}
\algnewcommand\Param{\item[\algorithmicparam]}

\usepgfplotslibrary{fillbetween}
\usetikzlibrary{patterns}

\newcommand{\echannel}[0]{\mathcal{E}}

\newtheorem{thm}{Theorem}

\renewcommand{\epsilon}{\varepsilon}

\DeclareMathOperator{\Ber}{\mathrm{Ber}}

\newcommand{\ISD}{\mathsf{overlay}}

\newcommand{\weak}{\mathsf{weak}}
\newcommand{\dis}{\mathsf{dis}}
\newcommand{\mypartial}{\mathsf{partial}}

\newcommand{\KeyGen}{\mathsf{KeyGen}}
\newcommand{\Encode}{\mathsf{Encode}}
\newcommand{\Decode}{\mathsf{Decode}}

\newcommand{\negl}{\mathsf{negl}}

\usepackage{tikz}
\usepackage{ragged2e}
\usetikzlibrary{positioning}

\usepackage{makecell}

\usepackage{tikz}
\usepackage{amsfonts}
\usepackage{xspace} 
\usepackage{amsmath}

\usepackage{mathrsfs}
\usepackage{amssymb}
\usepackage{amsmath}
\usepackage{amsthm}
\usepackage{multirow}
\usepackage{makecell}
\usepackage{caption}
\usepackage{subcaption}
\usepackage{float}
\usepackage{booktabs}
\usepackage{balance}
\usepackage{enumitem}
\usepackage{xcolor}
\usepackage{cleveref}
\crefname{algorithm}{Algorithm}{Algorithms} 
\Crefname{algorithm}{Algorithm}{Algorithms} 
\usepackage{url}
\usepackage{array}
\usepackage{diagbox}
\usepackage{url}
\usepackage{booktabs}
\usepackage{amssymb}
\usepackage{bbding}
\usepackage[most]{tcolorbox}
\usepackage{pifont}
\usepackage{utfsym}
\usepackage{fontawesome}
\usepackage{dutchcal}
\usepackage{xspace}
\usepackage{bm}
\usepackage{placeins}
\usepackage{colortbl}
\usepackage{multirow}
\usepackage[normalem]{ulem}
\useunder{\uline}{\ul}{}
\usepackage{booktabs}
\usepackage{array}
\usepackage{xfrac}
\usepackage{subcaption}
\newcommand{\mypara}[1]{\smallskip\noindent{\bf {#1}.} \xspace}

\newtheorem{definition}{Definition}

\usepackage{upgreek}
\newcommand{\llm}{$\mathtt{\Pi}_{\rm LLM}$\xspace}

\newcommand{\gim}{$\mathtt{\Pi}_{\rm GIM}$\xspace}

\newtheorem{construction}{Construction}
\crefname{construction}{construction}{constructions}

\newtheorem{configuration}{Config}
\crefname{configuration}{configuration}{configurations}

\definecolor{applegreen}{rgb}{0.55, 0.71, 0.0}

\newtcolorbox{mybox}[1][]{
breakable, 
colback=gray!10!white, 
colframe=white, 
enhanced,
sharp corners, 
boxrule=0pt, 
left=2pt, right=2pt, top=2pt, bottom=2pt, 
before skip=5pt, after skip=5pt, 
}

\newtcolorbox{takeaway}[1][]{
breakable, 
colback=gray!5!white,
colframe=black,
coltitle=black,
fonttitle=\bfseries\itshape,
enhanced,
sharp corners=southwest,
drop shadow,
left=2pt, right=2pt, top=2pt, bottom=2pt, 
before skip=5pt, after skip=5pt, 
}

\usepackage{afterpage}

\begin{document}

\date{}

\title{\Large \bf Cryptanalysis of LDPC-Based Pseudorandom Error-Correcting Codes}

\author{
{\rm Tianrui Wang\textsuperscript{1}}\ \ \
{\rm Anyu Wang\textsuperscript{2,3,4}$^\dagger$}\ \ \
{\rm Tianshuo Cong\textsuperscript{5,6}$^\dagger$}\ \ \ \\
{\rm Delong Ran\textsuperscript{1}}\ \ \
{\rm Jinyuan Liu\textsuperscript{2}}\ \ \
{\rm Xiaoyun Wang\textsuperscript{2,3,4,7,8}}
\\
\textsuperscript{1}\textit{Institute for Network Sciences and Cyberspace, BNRist, Tsinghua University} \ \ \ \\
\textsuperscript{2}\textit{Institute for Advanced Study, Tsinghua University} \ \ \ 
\textsuperscript{3}\textit{Zhongguancun Laboratory, Beijing, China} \ \ \ \\
\textsuperscript{4}\textit{State Key Laboratory of Cryptography and Digital Economy Security, Tsinghua University} \ \ \ \\
\textsuperscript{5}\textit{School of Cryptologic Science and Engineering, Shandong University} \ \ \
\\
\textsuperscript{6}\textit{Shandong Key Laboratory of  Artificial Intelligence Security, Shandong University}
\\
\textsuperscript{7}\textit{National Financial Cryptography Research Center, Beijing, China}
\\
\textsuperscript{8}\textit{Shandong Institute of Blockchain, Shandong, China}
}

\maketitle

\newcommand\blfootnote[1]{%
\begingroup
\renewcommand\thefootnote{}\footnote{#1}%
\addtocounter{footnote}{-1}%
\endgroup
}
\blfootnote{$^\dagger$ Corresponding authors: Anyu Wang (anyuwang@tsinghua.edu.cn) and Tianshuo Cong (tianshuo.cong@sdu.edu.cn)}

\begin{abstract}
Pseudorandom error-correcting codes (PRCs), a novel cryptographic primitive recently proposed at CRYPTO 2024, are primarily applied in undetectable watermarking schemes for large generative models. However, the security of PRCs has not yet been systematically analyzed. To fill this gap, we present the first cryptanalysis of PRCs. Specifically, focusing on LDPC-PRC, the only known practical instantiation of PRCs, we propose three novel attacks that challenge its undetectability and robustness. To rigorously demonstrate the practical threat, we analyze the concrete attack complexity under realistic parameters and validate the attack effectiveness on both real-world large language models and generative image models, including DeepSeek and Stable Diffusion. Our analysis shows that the claimed security guarantees of LDPC-PRC are undermined across all practically feasible regimes. For example, our attacks can detect the presence of a watermark with overwhelming probability at a cost of $2^{22}$ operations. Beyond attacks, we further propose three defenses: parameter recommendation, implementation suggestion, and a revised key generation function. However, PRC-based watermarking schemes still fail to achieve 128-bit security due to inherent constraints of large generative models, such as the maximum output length of large language models. Overall, our work clarifies the concrete security limits of PRCs in real-world watermarking applications.

\end{abstract}

\section{Introduction}
\label{sect/intro}

The rapid growth of generative Artificial Intelligence (AI) models, such as GPT-4o~\cite{openai2024gpt4ocard} and Stable Diffusion~\cite{Rombach_2022_CVPR}, has made AI-Generated Content (AIGC) ubiquitous~\cite{aimedia,aiad}. 
However, this widespread adoption is accompanied by substantial misuse risks~\cite{ALBUSAIDI2024100630,Lucchi_2024}. 
Consequently, reliable identification and provenance tracking of AIGC have become critical for AI governance, with content watermarking standing out as one of the most promising technical solutions~\cite{sok_zhao_sp25}.

Notably, a compelling foundation for AIGC watermarks is the recently developed cryptographic primitive of {Pseudorandom Error-correcting Codes} (PRCs)~\cite{christ2024pseudorandom}. 
A PRC scheme comprises three algorithms $(\KeyGen, \Encode, \Decode)$. 
The $\KeyGen$ algorithm generates a public-secret key pair. 
Any party can use the public key with $\Encode$ to embed a message into a codeword, while only the holder of the corresponding secret key can apply $\Decode$ to perform verification.
Crucially, PRCs are defined by two core security properties: \textit{robustness}, which guarantees that a valid codeword, even after limited distortion, will decode to the original message; and \textit{undetectability}, which ensures that codewords are computationally indistinguishable from random vectors without the secret key. 
These properties capture the essential requirements for a practical watermark by ensuring it resists removal and detection by adversaries, making PRCs particularly suitable for constructing stealthy and robust watermarks for both Large Language Models (LLMs) and Generative Image Models (GIMs).

To instantiate these properties, Christ and Gunn proposed \ldpc~\cite{christ2024pseudorandom}, an instantiation of PRCs based on Low-Density Parity-Check (LDPC) codes.
\ldpc offers provable security under standard assumptions, namely the hardness of the Learning Parity with Noise (LPN) and planted XOR problems. 
These theoretical foundations provide PRC-based watermarking shemes with significant advantages over other heuristic alternatives~\cite{zhao2024provable,DBLP:journals/tmlr/KuditipudiTHL24,DBLP:conf/colt/ChristGZ24,10.62056/ahmpdkp10,Aar22,kirchenbauer2024watermarklargelanguagemodels}.

Despite strong theoretical guarantees, the concrete security of PRCs under practical parameter settings remains largely unexplored. 
In particular, existing PRCs including \ldpc are vulnerable to quasipolynomial-time distinguishing attacks~\cite{christ2025improved}, indicating that their practical undetectability may be significantly weaker than expected.
Although undetectability was originally introduced to preserve the quality of generated content, in realistic watermarking scenarios, breaking PRC undetectability can directly lead to unauthorized watermark detection, information leakage, or further facilitate subsequent watermark removal and forgery attacks.
Likewise, compromising PRC robustness directly leads to targeted watermark removal attacks.
Accordingly, a rigorous security evaluation of PRCs under practical parameters is urgently needed.

\subsection{Our Contributions}
This work presents the first concrete security analysis of PRCs, bridging the gap between asymptotic security guarantees and practical deployment requirements.
Since \ldpc is the only feasible instantiation of PRC to date, our analysis focuses on designing attacks to challenge the undetectability and robustness of \ldpc, consequently facilitating the process of detecting and removing the watermarks of \ldpc-based watermarking schemes.

\mypara{Novel Cryptographic Attacks}
We propose three attacks against \ldpc, with goals ranging from watermark detection to watermark removal.
\Cref{tab:summ_attack_intro} summarizes the theoretical complexity of each attack.

\begin{itemize}

\item \textbf{Attack-I: Partial Secret Key Recovery:}
This attack aims to recover part of the secret key from the public key, enabling an adversary to emulate the $\Decode()$ procedure of \ldpc. 
This gives the adversary a constant distinguishing advantage greater than $\frac{1}{2}$ between codeword and random vector, allowing watermark detection with overwhelming probability given sufficient samples.

\item \textbf{Attack-II: Weak Key Distinguisher:}
This attack exploits an implementation-level vulnerability in the key generation procedure. 
We observe that in instantiated \ldpc-based watermarking schemes, including those for LLMs~\cite{christ2024pseudorandom} and GIMs~\cite{DBLP:conf/iclr/GunnZS25}, public keys may exhibit statistical non-uniformity with non-negligible probability. 
We demonstrate that an adversary can leverage this non-uniformity to detect the presence of a watermark with overwhelming probability.

\item \textbf{Attack-III: Noise Overlay Attack:}
This attack exploits the bounded noise tolerance of the \ldpc's $\Decode$ algorithm. 
By recovering the original Bernoulli noise vector $\vec{e}$, an adversary can construct a carefully designed overlay vector $\vec{e}'$ of a fixed noise weight such that $\vec{e} + \vec{e}'$ exceeds the decoder's correction threshold, whereas a random $\vec{e}'$ with the same weight would not.
This results in watermark removal without violating the practical constraints of LLM or GIM applications.

\end{itemize}


\mypara{Complexity for Concrete Parameter Configurations}
We further analyze the concrete complexities of our attacks by incorporating specific parameter settings for \ldpc-based LLM watermarking scheme (denoted as \llm~\cite{christ2024pseudorandom}) and GIM watermarking scheme (denoted as \gim~\cite{DBLP:conf/iclr/GunnZS25}). 
For Attack-I, our analysis shows that the complexities are sub-security-level for all parameter choices, where the actual security of \gim remains \textit{below $46$ bits}.  
Attack-II further exposes vulnerabilities resulting from a non-negligible frequency of weak keys across all parameter choices.  
For example, weak keys appear with a nearly $100\%$ probability when $t=3,4$ for \llm, under which a distinguisher can be constructed in polynomial time. 
For Attack-III, we prove that the security level of \llm \textit{cannot exceed $75$ bits} for any parameter choice with $n \le 2^{17}$ under the threat of noise overlay attack.  
Meanwhile, the security level of \gim \textit{is below $60$ bits} with a much lower concrete noise rate.

\begin{table}[t]
\centering
\caption{The summary of the attack complexities. $n$ is the code length of \ldpc, $g$ is the column dimension of the public key $\mat{G}$, $t$ is the row weight of the secret key $\mat{P}$. 
The $\alpha$ in Attack-I is calculated in Theorem~\ref{thm:partial_recovery}. For Attack-II, $\pr{P}$ is the probability that a weak key exists. 
The $\epsilon$ in Attack-III corresponds to the error-correcting capability.}
\setlength{\tabcolsep}{7pt}
\begin{tabular}{c|c|c}
\toprule
   & Goal                                                                                    & \begin{tabular}[c]{@{}c@{}}Time\\ Complexity\end{tabular}                                      \\ \midrule
Attack-I  & \begin{tabular}[c]{@{}c@{}}Against Undetectability\\ (Watermark Detection)\end{tabular} & \begin{tabular}[c]{@{}c@{}}$g \cdot \alpha \cdot \binom{n/2}{\lceil t/2 \rceil}$\\ Theorem~\ref{thm:partial_recovery}\end{tabular}                                                                               \\ \midrule
Attack-II & \begin{tabular}[c]{@{}c@{}}Against Undetectability\\ (Watermark Detection)\end{tabular} & \begin{tabular}[c]{@{}c@{}}$\pr{P}^{-1} \cdot n \cdot g$\\ Theorem~\ref{thm:weakey_dis}\end{tabular}                                                                                                                                         \\ \midrule
Attack-III  & \begin{tabular}[c]{@{}c@{}}Against Robustness\\ (Watermark Removal)\end{tabular}        & \begin{tabular}[c]{@{}c@{}}$(\frac{1}{2}+\epsilon)^{-g} \cdot n^{3}$\\ Theorem~\ref{thm:error_recover}\end{tabular}                                                                                                    \\ \bottomrule
\end{tabular}
\label{tab:summ_attack_intro}
\end{table}

\mypara{Evaluation on Real-World Applications}
We choose DeepSeek~\cite{deepseekai2025deepseekr1incentivizingreasoningcapability} and Stable Diffusion~\cite{Rombach_2022_CVPR} as real-world LLM and GIM to implement \llm and \gim, respectively.
To the best of our knowledge, this is the first real-world implementation for \llm.
However, we validate that \llm is impractical for embedding watermarks into LLMs as the entropy of output tokens is too low under reasonable text qualities, rendering the watermark undetectable by the original decoder.
Furthermore, evaluation results demonstrate that both the proposed Attack-I and Attack-II can achieve a 100\% attack success rate against \llm and \gim under certain parameter choices (refer to Table~\ref{tab:concrete_experiment_llm} and Table~\ref{tab:concrete_experiment_gim}).

\mypara{Mitigation Methods}
We discuss strategies to mitigate the threats posed by the proposed attacks. 
For Attack-I, we suggest that selecting appropriate parameters is sufficient to make recovering any part of the secret key computationally difficult. 
For Attack-II, we propose a revised key generation algorithm that significantly reduces the probability of generating weak keys while maintaining the randomness of the \ldpc keys. 
Regarding Attack-III, the most effective defense is to increase the code length $n$ to raise the attack's complexity. 
However, achieving a high security level, e.g., $128$-bit, may necessitate an unreasonably large $n$.
Our analysis shows that $n > 2^{24}$ is a necessary condition for $128$-bit security across all choices of $t$. 
This requirement far exceeds the maximum token length of state-of-the-art LLMs (only $2^{15}$).
Consequently, defending against Attack-III through parameter selection alone remains difficult for practical LLM implementation. 
Additionally, We further suggest a modification to the GIM implementation, which mitigates the vulnerability caused by its low-weight noise parameter.

\mypara{Takeaways}
Our analysis reveals a practical limitation of current PRC-based LLM watermarking schemes, where preventing practical watermark-removal attacks requires an impractically large code length $n$, far beyond the output length supported by current LLMs.
In contrast, PRC-based image watermarking may still admit practical instantiations with acceptable perceptual quality, although their undetectability guarantees can be weakened in practice.

\section{Related Works}
\label{sect/related}
\mypara{Theoretical Studies on PRCs}
Due to the promising potential of PRCs, besides studying how to apply PRCs to AIGC watermarking~\cite{christ2024pseudorandom,DBLP:conf/iclr/GunnZS25}, several fundamental theoretical studies on PRCs have recently emerged, such as proving CCA-security based on the hardness of LPN~\cite{alrabiah2024ideal}, discussing unconditionally indistinguishability against space-bounded adversaries~\cite{ghentiyala2024new}, and performing a black-box reduction from one-way functions with sub-constant error rates \cite{garg2025blackboxcryptouselesspseudorandom}. 
Unfortunately, these studies have not thoroughly explored the security of PRCs. 
Note that although \cite{christ2024pseudorandom,alrabiah2024ideal} provide a theoretical evaluation of PRC's robustness, it does not conduct a security assessment in conjunction with the concrete parameters.

\mypara{Watermarking Schemes for AIGC}
AIGC watermarks can typically be classified into two categories: in-processing schemes and post-processing schemes.
Post-processing schemes embed imperceptible yet detectable signals after text generation~\cite{sato2023embarrassinglysimpletextwatermarks,liu2024survey} or image generation~\cite{deb2012combined,zhu2018hidden,tancik2020stegastamp}, inevitably resulting in content quality degradation.
Instead, both \llm and \gim belong to the in-processing approach~\cite{kirchenbauer2024on,zhao2024provable,wen2023tree,yang2024gaussian,fernandez2023stable}, wherein watermark insertion is merged with the process of generating content.
With a pseudorandom codeword embedded, the distribution of the watermarked output texts of \llm is proven to be identical to that of the original texts, while \gim utilizes a watermarked latent that maintains a Gaussian distribution to ensure watermark invisibility.

\mypara{Attacks against AIGC Watermarks}
To remove watermarks, mainstream methods typically apply heuristic perturbations to the AIGC, such as paraphrasing~\cite{DBLP:conf/nips/KrishnaSKWI23,pan2024markllm} or inserting perturbations~\cite{kirchenbauer2024on,jovanovic2024watermark} for the generated texts, or using processing operations~\cite{tallam2025removingwatermarkspartialregeneration,DBLP:conf/icml/AnDRAXDZMWGH24} such as photometric distortions for images.
Unlike these empirical removal attacks, our cryptographic attacks provide provable success rates and well-defined complexity.
For PRC watermark detection, Gunn et al.~\cite{DBLP:conf/iclr/GunnZS25} trained a ResNet-18 classifier~\cite{He_2016_CVPR} and reported that it fails to distinguish \ldpc-based watermarked images from non-watermarked ones.
However, we propose two different attacks to challenge the undetectability of \gim.

\section{Preliminaries}
\label{sect/pre}
\subsection{Pseudorandom Error-correcting Codes}
\label{subsec:brief_descrip_prc}
We begin by introducing the formal definition of PRC and its instantiation based on LDPC codes. 
Our focus is on public-key PRCs due to their broader applicability, as they allow any party to embed a watermark while reserving detection capability for a regulator holding the secret key. 
In contrast, secret-key PRCs enable any party capable of embedding a watermark to also detect it locally, which makes forgery or removal attacks easier for malicious users.

\subsubsection{Definition of PRCs}



\begin{definition}[Zero-bit public-key PRCs~\cite{christ2024pseudorandom}]
\label{def:prc}
Let $\Sigma$ be a fixed alphabet, $\echannel: \Sigma^* \rightarrow \Sigma^*$ be a noisy channel modeled by adding a noise vector with a Bernoulli distribution, and $k(\lambda) \in \mathbb{N}$ be the length of the message under any fixed $\lambda$. A PRC is a triple of polynomial-time randomized algorithms $(\KeyGen, \Encode, \Decode)$ that satisfy
\begin{itemize}

\item (\textbf{Robustness}) For any message $\mathsf{m} \in \Sigma^{k(\lambda)}$ and any $\lambda \in \mathbb{N}$, let $\vec{x} \leftarrow \Encode(1^\lambda,\mathsf{pk},\mathsf{m})$, $\vec{x}' \leftarrow \echannel(\vec{x}) $ then
\[ 
\mathop{\Pr}\limits_{(\mathsf{sk},\mathsf{pk}) \leftarrow \KeyGen(1^\lambda)}[\mathsf{Decode}(1^\lambda,\mathsf{sk},\vec{x}')=\mathsf{m}] \ge 1-\negl(\lambda).
\]
\item (\textbf{Soundness}) For any fixed $\vec{c} \in \Sigma^*$,
\[ 
\mathop{\Pr}\limits_{(\mathsf{sk},\mathsf{pk}) \leftarrow \KeyGen(1^\lambda)}[\Decode(1^\lambda,\mathsf{sk},\vec{c})=\perp ] \ge 1-\negl(\lambda).
\]
\item (\textbf{Undetectability}) For any polynomial-time adversary $\mathcal{A}$,
let
$
\Pr[\mathsf{E}]= {\Pr}[\mathcal{A}^{\Encode(1^\lambda,\mathsf{pk},\cdot)}(1^\lambda,\mathsf{pk})=1]
$
and
$\Pr[\mathsf{U}] = {\Pr}[\mathcal{A}^{\mathcal{U}}(1^\lambda,\mathsf{pk})=1] $, then 
\[
| \mathop{\Pr}\limits_{(\mathsf{sk},\mathsf{pk}) \leftarrow \KeyGen(1^\lambda)}[\mathsf{E}] - \mathop{\Pr}\limits_{(\mathsf{sk},\mathsf{pk}) \leftarrow \KeyGen(1^\lambda) \atop \mathcal{U}}[\mathsf{U}] |
\le \negl(\lambda).      
\]
where $\mathcal{A}^{\mathcal{U}}$ denotes an adversary with access to a uniform random oracle along with previous queries.


\end{itemize}
\end{definition}

\noindent Definition~\ref{def:prc} describes zero-bit PRCs, where the encoded message $\mathsf{m}$ is a single fixed message and can therefore be omitted in the following discussions. 
Multi-bit PRCs are constructed upon this zero-bit primitive~\cite{christ2024pseudorandom}. 
In this paper, we discuss attacks specifically against zero-bit PRCs, and these attacks naturally extend to the multi-bit setting. 
Furthermore, as noted in~\cite{christ2024pseudorandom}, the soundness property is primarily a technical condition ensuring non-triviality.
Consequently, our attacks do not target this property.


\mypara{Relationship to Stronger Robustness Definitions}
In this work, we focus on the basic robustness notion in Definition~\ref{def:prc}, where the adversary perturbs through a substitution noise. 
Other stronger robustness definitions have also been considered. 
For example, adaptive robustness in~\cite{alrabiah2024ideal} grants the adversary interactive query access of the decoder. 
Substring robustness in~\cite{christ2024pseudorandom} further allows deletion-based cropping attacks in addition to the substitution perturbations already covered by Definition~\ref{def:prc}.
Therefore, these robustness notions assume strictly stronger adversarial capabilities.

Our attack in~\Cref{subsec:msg_recovery} constructs perturbations by exploiting the underlying error structure of the scheme. 
Since adaptive robustness and substring robustness only extend the adversarial capabilities beyond those in Definition~\ref{def:prc}, the same attack strategy can be directly applied to break these stronger robustness notions as well.

\subsubsection{Instantiation of PRCs}

\begin{construction}[\ldpc~\cite{christ2024pseudorandom}]
\label{def:pub_prc}

Based on LDPC codes, an instantiation of PRC named \ldpc can be constructed as follows.
\begin{itemize}
\item $\mathsf{KeyGen}(1^{\lambda})${\rm :} 
\begin{enumerate}
\item Sample a random matrix $\mat{P} \in \GF{2}^{r \times n}$ where each row of $\mat{P}$ has weight $t$.
\item Sample a random matrix $\mat{G} \in \GF{2}^{n \times g}$ where $\mat{P}\mat{G}=\mat{0}$. 
\item Sample a random vector $\mat{z} \in \GF{2}^{n}$.
\item Output public key $\mathsf{pk}=(\mat{G},\mat{z})$ and secret key $\mathsf{sk}=(\mat{P},\mat{z})$.
\end{enumerate}

\item $\mathsf{Encode}(1^\lambda,(\mat{G},\mat{z}))${\rm :}
\begin{enumerate}
\item Sample a random vector $\mat{s} \in \GF{2}^{g}$ and $\mat{e} \leftarrow \Ber(n,\omega)$.
\item Output the codeword $\mat{x} \leftarrow \mat{G}\mat{s}+\mat{e}+\mat{z}$.
\end{enumerate}

\item $\mathsf{Decode}(1^\lambda,(\mat{P},\mat{z}),\mat{x})${\rm :}
If $\wt{\mat{P}(\mat{x}+\mat{z})} \le (\frac{1}{2}-r^{-1/4})\cdot r$, output 1; otherwise output $\perp$.
\end{itemize}
\end{construction}


\noindent For clarity, we list the key mathematical notations for \ldpc in~\Cref{tab:notations} (Appendix~\ref{supp/llm_gim_watermark}).

As defined in~\Cref{def:pub_prc}, the \ldpc scheme begins by generating a key pair. 
The secret key $\mat{P}$ is an LDPC parity-check matrix, while the public key $\mat{G}$ is a matrix whose column space is the dual subspace of $\mat{P}$. 
The random vector $\vec{z}$ serves as a one-time pad.
Without it, the zero vector would always decode to $1$, violating the soundness property.

Under the planted XOR assumption, the public key $\mat{G}$ does not leak information about the secret key $\mat{P}$. 
Furthermore, the hardness of the LPN problem ensures that the encoded message $\mat{G}\vec{s} + \vec{e}$ is computationally indistinguishable from a uniform random vector.

During decoding, the Hamming weight of $\mat{P}(\vec{x}+\vec{z}) = \mat{P}(\mat{G}\vec{s} + \vec{e}) = \mat{P}\vec{e}$ is significantly lower than $\frac{r}{2}$ due to the sparsity of both $\mat{P}$ and $\vec{e}$. 
In contrast, the product of $\mat{P}$ with a random vector yields a Hamming weight close to $\frac{r}{2}$. 
In the zero-bit setting, this decoding procedure is equivalent to watermark detection.
Consequently, the watermark can be detected with overwhelmingly high probability by comparing the Hamming weight to an appropriately chosen threshold.

\subsection{Large Generative Models}
As PRCs are primarily applied in undetectable watermarking schemes for large generative models, this part introduces the fundamental generative principles of such models to facilitate a more comprehensive understanding of watermarking schemes.
The formal definitions of LLMs and GIMs are as follows.

\begin{definition}[$\mathsf{LLM}$]
A large language model $\mathsf{LLM}$ over a token set $\mathcal{T}$ is a probabilistic algorithm that takes a prompt $\mathsf{prompt}\in \mathcal{T}^*$ as input and generates succeeding text $T\in \mathcal{T}^\mathsf{len}$.
The generation process of $\mathsf{LLM}$ is detailed in~\Cref{alg:LLM}.
\end{definition}

\begin{definition}[$\mathsf{GIM}$]
A generative image model $\mathsf{GIM}$ is a probabilistic algorithm that takes a prompt $\mathsf{prompt}\in \mathcal{T}^*$ as input and outputs an image $\mathsf{Img} \in [0,255]^{c\times w \times h}$.  
The generation process of $\mathsf{GIM}$ is detailed in~\Cref{alg:GIM}.
\end{definition}

For LLMs, the model takes a prompt as input and generates text in a token-by-token manner.
At each step, the model receives the initial prompt combined with the previously generated tokens as input, producing a probability distribution over the vocabulary from which the subsequent token is sampled accordingly.
Watermarks can be embedded by influencing this probability distribution using a vector generated by \ldpc.

For GIMs, the model first samples an initial latent from a Gaussian distribution and then progressively denoises it
using a \texttt{U-net} conditioned on the input prompt.
Finally, an image is generated via an autoencoder decoder \texttt{D-net}.
Watermarks are embedded by selecting the signs of the initial latent according to the binary string generated by \ldpc.

\subsection{PRC-Based Watermarking Schemes}
\label{sec:llm-scheme}

We formalize the key properties of a watermarking scheme for generative models in~\Cref{def:watermarking_scheme}.
Informally, \textit{robustness} captures the difficulty of removing the watermark under bounded disturbance, \textit{soundness} requires a negligible false positive rate, and \textit{undetectability} means that no efficient adversary can detect the presence of a watermark with non-negligible advantage.

\begin{definition}[Watermarking Scheme]
Let $\Sigma$ be a fixed alphabet and $\echannel: \Sigma^* \rightarrow \Sigma^*$ be a channel with the security parameter $\lambda$,
a watermarking scheme is a tuple of polynomial-time algorithms $(\mathsf{Setup},\mathsf{Wat},\mathsf{Detect})$ that satisfy 
\begin{itemize}
    \item (Robustness) For any content $x \leftarrow \mathsf{Wat}_{\mathsf{pk}}(\mathsf{prompt})$,
    \[ 
\mathop{\Pr}\limits_{(\mathsf{sk},\mathsf{pk}) \leftarrow \mathsf{Setup}(1^\lambda)}[\mathsf{Detect}_{\mathsf{sk}}(\echannel(x))=\mathsf{true}] \ge 1-\negl(\lambda).
\]
    \item (Soundness) For any fixed $c \in \Sigma^*$,
    \[ \mathop{\Pr}\limits_{(\mathsf{sk},\mathsf{pk}) \leftarrow \mathsf{Setup}(1^\lambda)}[\mathsf{Detect}_{\mathsf{sk}}(c)=\mathsf{false}] \ge 1-\negl(\lambda).
    \]
    \item (Undetectability) For any polynomial-time adversary $\mathcal{A}$ mimics $\mathsf{Detect}$, let $\mathsf{Uni}$ be an oracle that outputs random vector in $\Sigma^*$, and $
    x \xleftarrow{\$}
\mathsf{Uni}() $ with probability $\frac{1}{2}$ while $ x \xleftarrow{\$}
\mathsf{Wat}_{\mathsf{pk}}(\mathsf{prompt})      \text{ with probability } \frac{1}{2}
$, then
    \[ | \mathop{\Pr}\limits_{(\mathsf{sk},\mathsf{pk}) \leftarrow \mathsf{Setup}(1^\lambda)}[\mathsf{Detect}_{\mathsf{sk}}(x)=\mathcal{A}_{\mathsf{pk}}(x)] - \frac{1}{2} | \le \negl(\lambda),\]

\end{itemize}
\label{def:watermarking_scheme}
\end{definition}


\begin{algorithm}[t]
\caption{Key Generation Algorithm of \llm~\cite{christ2024pseudorandom}}
\label{alg:watermark_keygeneration}
\begin{algorithmic}[1]

\Require {Parameters $n$, $r=0.99n$, $g$, and $t$}
\Ensure Matrices $\mat{P}$ and $\mat{G}$ for watermark generation
\State Sample a uniformly random matrix $\mat{G_0} \leftarrow \mathbb{F}_2^{0.01n \times g}$
\State Initialize an empty list for rows of $\mat{P}$

\ForAll{$i \in [0.99\,n]$}
    \State Sample a random $(t-1)$-sparse vector $\vec{s}_i \in \mathbb{F}_2^{0.01n}$
    \State Compute $\mat{G}_i$ by appending the row $\vec{s}_i^T \mat{G}_0$ to $\mat{G}_{i-1}$:
    \[
    \mat{G}_i = 
    \begin{bmatrix}
    \mat{G}_{i-1} \\
    \vec{s}_i^T \mat{G}_0
    \end{bmatrix}
    \]
    \State Construct $\vec{s}_i' = [\vec{s}_i^T, 0^{i-1}, 1, 0^{0.99n-i}]$
    \State Append $\vec{s}_i'$ to the list of rows for $\mat{P}$
\EndFor

\State Let $\mat{P}$ be the matrix whose rows are $\vec{s}_1', \ldots, \vec{s}_{0.99n}'$
\State Let $\mat{G} = \mat{G}_{0.99n}$
\State Sample a random permutation ${\Pi} \in \GF{2}^{n \times n}$ and let $\mat{P} \leftarrow \mat{P}\Pi^{-1}, \mat{G} \leftarrow \Pi\mat{ G} $

\State \Return $(\mat{P}, \mat{G})$

\end{algorithmic}
\end{algorithm}

\subsubsection{Watermarking Scheme for LLMs}

In~\cite{christ2024pseudorandom}, Christ and Gunn proposed an approach to design a watermarking scheme for LLMs (denoted as \llm).
The specific construction is presented in~\Cref{construction_llm}.

\begin{construction}[\llm~\cite{christ2024pseudorandom}]
\label{construction_llm}
A watermarking scheme \llm for $\mathsf{LLM}$ over $\mathcal{T}$ is a tuple of polynomial-time algorithms \llm=$(\mathsf{Setup},\mathsf{Wat},\mathsf{Detect})$ where
\begin{itemize}
    \item $\mathsf{Setup}(1^\lambda)$ outputs $(\mathsf{sk},\mathsf{pk}) \leftarrow \mathsf{LDPC\text{-}PRC}.\KeyGen(1^\lambda)$.
    The matrices are sampled through~\Cref{alg:watermark_keygeneration}.
    
    \item $\mathsf{Wat}_{\mathsf{pk}}(\mathsf{prompt})$ is a randomized algorithm that generates watermarked responses.
    The entire response generation procedure is similar to~\Cref{alg:LLM}, except for the token sampling process (Line 5).
    Specifically, \llm first outputs a codeword $\vec{x} \leftarrow \mathsf{LDPC\text{-}PRC}.\mathsf{Encode}(1^\lambda,\mathsf{pk})$, where $|\vec x |=\mathsf{len}*\lceil \log_2|\mathcal{T}|\rceil$.
    Then, this codeword is split as $\mathsf{len}$ trunks and guides the token sampling process in a bit-wise manner through~\Cref{alg:TokenSample}.
    
    \item $\mathsf{Detect}_{\mathsf{sk}}({T})$ is an algorithm that outputs $\mathsf{true}$ or $\mathsf{false}$.
    \llm first extracts the codeword $\vec x'$ following~\Cref{alg:TextInv}, then runs $\mathsf{LDPC\text{-}PRC}.\Decode(1^\lambda,\mathsf{sk},\vec x')$.
    If the result is $1$, this algorithm outputs $\mathsf{true}$, indicating the watermark is present; otherwise outputs $\mathsf{false}$.
\end{itemize}
\end{construction}

\noindent 
Overall, \llm embeds an \ldpc codeword into the output text by biasing the token sampling probabilities according to the codeword bits.
During generation, the codeword is partitioned into fixed-length segments, and each segment influences the sampling distribution of one token at the bit level.
On the detection side, the embedded bits are reconstructed from the generated text and verified by applying the \ldpc decoding algorithm with the secret key.
However,~\Cref{construction_llm} is only a theoretical construction.
In~\Cref{subsec:realworld_llm}, we will analyze its practical utility and show that the output text quality is severely compromised, since the error-correcting capability of \llm is too limited to support reliable detection under low token entropy.
We further provide the first implementation of \llm on a real-world LLM, empirically demonstrating its impracticality.

\subsubsection{Watermarking Scheme for GIMs}

Gunn et al.~\cite{DBLP:conf/iclr/GunnZS25} further extended \ldpc to the design of GIM watermarking schemes (denoted as \gim).
A detailed description of \gim is provided in~\Cref{construction_gim}.

\begin{construction}[\gim~\cite{DBLP:conf/iclr/GunnZS25}]
\label{construction_gim}
A watermarking scheme \gim for $\mathsf{GIM}$ is a tuple of polynomial-time algorithms $(\mathsf{Setup},\mathsf{Wat},\mathsf{Detect},\mathsf{Decode'})$ where
\begin{itemize}
    \item $\mathsf{Setup}(1^\lambda)$ outputs $(\mathsf{sk},\mathsf{pk}) \leftarrow \mathsf{LDPC\text{-}PRC}.\KeyGen(1^\lambda) $.
    The matrices are sampled in a way similar to~\Cref {alg:watermark_keygeneration}, which will be discussed later. 
    \item $\mathsf{Wat}_{\mathsf{pk}}(\mathsf{prompt})$ is a randomized algorithm that outputs a watermarked image $\widetilde{\mathsf{Img}}$.
    To this end, \gim uses a watermarked latent $\tilde{\vec{y}}$ as $\vec{y}^s$ of~\Cref{alg:GIM} to generate images.
    The generation process of $\tilde{\vec{y}}$ is 
    \[
    \tilde{\vec{y}}=(\tilde{y}_1,\cdots,\tilde{y}_n) \quad \text{where} \quad \tilde{y}_i =(1-2x_i)\,|y_i|,
    \]  
    where $\vec{y}=(y_1,\cdots,y_n)\sim \mathcal{N}(\vec 0,{\bf I}_n)$ and $\vec{x} \leftarrow \mathsf{LDPC\text{-}PRC}.\mathsf{Encode}(1^\lambda,\mathsf{pk})$.



    \item $\mathsf{Detect}_{\mathsf{sk}}(\mathsf{Img})$ is an algorithm that outputs whether $\mathsf{Img}$ is watermarked.
    \gim first extract a codeword $\vec{x}'$ from $\mathsf{Img}$ through~\Cref{alg:GIMInv}, and then run \ldpc.$\Decode(1^\lambda,\mathsf{sk},\vec x')$.
    If the result is $1$, $\mathsf{Detect}$ outputs $\mathsf{true}$, and $\mathsf{false}$ vise versa.
\end{itemize}
\end{construction}

\noindent We highlight that there are two main differences between \gim and \llm.
First, \gim initializes the key pair by generating a sparse vector $\vec{s}_i \in \GF{2}^{(n-r+i)}$ instead of $\GF{2}^{n-r}$ (Line 4 of~\Cref{alg:watermark_keygeneration}).
Second, in order to embed extra messages and estimate the false positive rate, \gim increases the column number of $\mat{G}$ from $g$ to $k$ by adding message bits and parity-check bits, which could be seen as a multi-bit PRC with its own implementation.
Besides detecting the watermark, \gim also introduces a $\Decode'()$ function to recover the original message $\vec{s}$ with the belief propagation algorithm, leveraging the sparsity of the secret key $\mat{P}$.
However, we will demonstrate that such a modification can cause vulnerability, as shown in~\Cref{subsec:defense2}. 
The complete construction and implementation details of \gim can be found in~\cite{DBLP:conf/iclr/GunnZS25}.

\section{Attacks Against \ldpc}
\label{sect/adaptive_attack}
In this section, we propose three attacks and provide analyses of their complexities.

\subsection{Attack-I: Partial Secret Key Recovery}
\label{sec:partial_sk_recovery}

Given a public key $\mat{G}\in \GF{2}^{n \times g}$, Attack-I aims to recover a set of row vectors $\{\vec{v}^{(j)}\} \subseteq \mathbb{F}^{n}_2$ of the secret key $\mat{P}$. 
Using these vectors, an adversary can distinguish whether a set of target vectors $\{\vec{x}^{(i)}\}$ is PRC-encoded by computing the distribution of $\langle\vec{v}^{(j)},\vec{x}^{(i)}+\vec{z}\rangle$. 
The attack procedure is shown in~\Cref{fig:workflow_key_recovery} and summarized in \Cref{alg:attack2_partial} of Appendix~\ref{supp/attack_algorithm}.

\begin{figure}[b]
\centering
\includegraphics[width=0.85\linewidth]{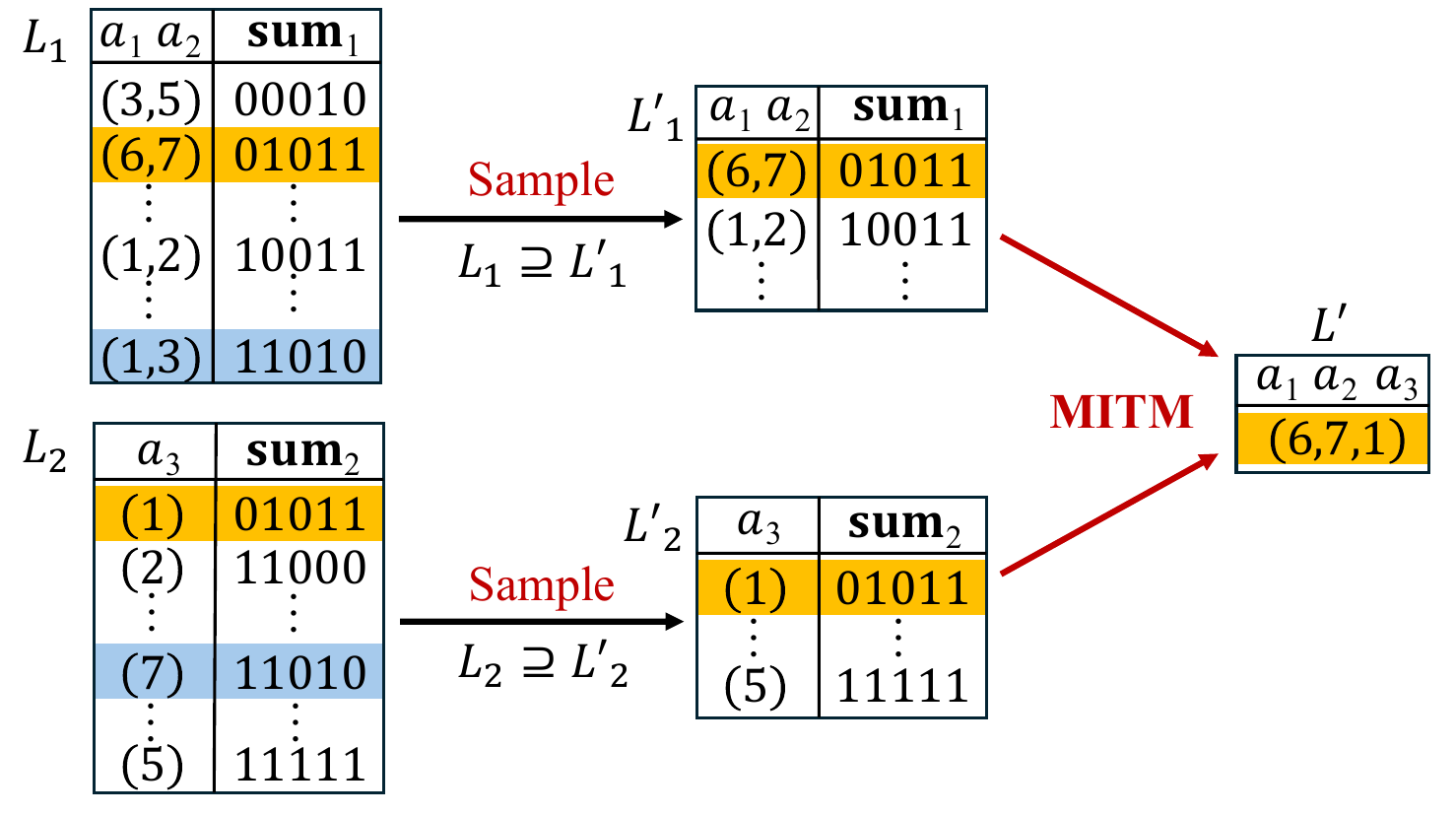}
\caption{Illustration of the Meet-in-the-Middle (MITM) Attack utilized in Partial Secret Key Recovery ($t=3$).
}
\label{fig:workflow_key_recovery}
\end{figure}

\subsubsection{Pipeline of Attack-I}

\mypara{Stage-I: Low-Weight Kernel Vectors Recovery}
Since Attack-I aims to compute a set of row vectors $\{\vec{v}^{(j)}\} \subseteq \mathbb{F}^{n}_2$ of the secret key $\mat{P}$ which satisfy $\vec{v}^{(j)} \mat{G} = \vec{0}$ and $\wt{\vec{v}^{(j)}} = t$, we adapt the Information Set Decoding (ISD) algorithm~\cite{becker2012decoding} to recover such a set of short vectors using a Meet-In-The-Middle (MITM) approach.

Let $n_1=\lceil n/2 \rceil$, $n_2=n-n_1$, $t_1=\lceil t/2 \rceil$, and $t_2=t-t_1$. 
For simplicity, we assume $n$ is even, so $n_1 = n_2 = n/2$. 
Let $\myset{L}$ be the set of vectors in $\mathbb{F}^{n}_2$ satisfying $\wt{\vec{v}_{1:n_1}}=t_1$ and $\wt{\vec{v}_{n_1+1:n}}=t_2$ for all $\vec{v} \in \myset{L}$.
The probability that a random vector of weight $t$ falls into $\myset{L}$ is $q=\sfrac{\binom{n_1}{t_1}\binom{n_2}{t_2}}{\binom{n}{t}}$. 
Thus, the expected number of rows of $\mat{P}$ in $\myset{L}$ is $r'=qr$, since each row of $\mat{P}$ is independently sampled at random.

To recover the expected $r'$ rows, we construct two sets $\myset{L}_1$ and $\myset{L}_2$, where $\myset{L}_1$ contains all $\binom{n}{t_1}$ sums of $t_1$ distinct rows from the upper half of $\mat{G}$ and $\myset{L}_2$ contains all $\binom{n}{t_2}$ possible sums of $t_2$ distinct rows from the bottom half of $\mat{G}$:
$\myset{L}_1 = \{ (a_1,...,a_{t_1},\vec{sum}_1) \mid 1\leq a_1 < \cdots < a_{t_1} \leq n, \vec{sum}_1=\sum_{i=1}^{t_1} \mat{G}_{a_i} \}$,
$\myset{L}_2 = \{ (a_{t_1+1},...,a_{t},\vec{sum}_2) \mid 1 \leq a_{t_1+1} < \cdots < a_{t} \leq n, \vec{sum}_2=\sum_{i=t_1+1}^{t} \mat{G}_{a_i} \},
$
where $\mat{G}_{a_i}$ denotes the $a_i$-th row of $\mat{G}$. 
We then employ the MERGE-JOIN algorithm~\cite{becker2012decoding} (a MITM implementation) to construct the target set
$$\myset{L}=\{(a_1,...a_{t_1},a_{t_1+1},...,a_t) \mid \vec{sum}_1 + \vec{sum}_2 = \vec{0} \},$$
where $(a_1,...,a_{t_1},\vec{sum}_1) \in \myset{L}_1$ and $(a_{t_1+1},...,a_{t},\vec{sum}_2) \in \myset{L}_2$. Thus, $\myset{L}$ stores the $r'$ rows of the secret key $\mat{P}$.

To improve efficiency, we reduce $\myset{L}_1$ and $\myset{L}_2$ to smaller sets, ensuring the expected number of recovered rows from $\mat{P}$ is $l < r'$:

\begin{itemize}
\item When $t$ is \textbf{even}, we randomly choose subsets $\myset{L}_i' \subset \myset{L}_i$ for $i=1,2$, with $|\myset{L}_1'|=\lceil \frac{1}{\sqrt{r'/l}} \binom{n_1}{t/2} \rceil$ and $|\myset{L}_2'|=\lceil \frac{1}{\sqrt{r'/l}} \binom{n_2}{t/2} \rceil$. Using MITM, we construct:
$$\myset{L}' = \{ (a_1,...,a_t) \mid \vec{sum_1} + \vec{sum_2} = \vec{0} \} \subseteq \myset{L},$$
where $(a_1,...,a_{t_1},\vec{sum_1}) \in \myset{L}_1'$ and $(a_{t_1+1},...,a_{t},\vec{sum_2}) \in \myset{L}_2'$. The expected number of rows of $\mat{P}$ in $\myset{L}'$ is $\frac{qr}{\sqrt{r'/l}\sqrt{r'/l}}=l$.

\item When $t$ is \textbf{odd}, we tailor $\myset{L}_1'$ and $\myset{L}_2'$ according to their sizes. Let $|\myset{L}_1'|=\alpha \binom{n/2}{t_1}$ and $|\myset{L}_2'|=\beta \binom{n/2}{t_2}$, where $\alpha=\sqrt{\frac{t_1}{(n/2-t_2)r'/l}}$ and $\beta=\frac{(n/2-t_2)\alpha}{t_1}$, ensuring $|\myset{L}_1'|=|{\myset{L}_2'}|$ and $\alpha \beta = l/r'$. Similarly, we construct $\myset{L}'$ via MITM, with the expected number of rows of $\mat{P}$ in $\myset{L}'$ equal to $\alpha\beta r' = l$.
\end{itemize}

\mypara{Stage-II: Distinguishing Vectors}
After recovering the rows $\{\vec{v}^{(1)},\cdots,\vec{v}^{(l)}\}$ of the secret key $\vec{P}$, we now can determine whether vectors $\{\vec{x}^{(1)},...,\vec{x}^{(m)}\}$ are \ldpc-encoded codewords or plain vectors. 
We assume these $m$ vectors are sampled independently. 

We first compute $\{\langle\vec{v}^{(j)},\vec{x}^{(i)}+\vec{z}\rangle\}_{1 \le j \le l, 1 \le i \le m}$ and then calculate the ratio of zeros in the result set. 
We can choose a threshold $\tau$ and conclude that the $m$ vectors are \ldpc-encoded if the ratio exceeds $\tau$, or plain vectors otherwise.

If the $m$ vectors are \ldpc-generated, i.e., $\vec{x}^{(i)}+\vec{z}=\mat{G}\vec{s}+\vec{e}^{(i)}$ with $\vec{e}^{(i)} \leftarrow \Ber (n,\omega)$, then
$$p \coloneqq \Pr[\langle\vec{v}^{(j)},\vec{x}^{(i)}+\vec{z}\rangle=0] = \sum_{\tilde{j}~\text{is}~\text{even}} \binom{t}{\tilde{j}} \cdot \omega^{\tilde{j}} \cdot (1-\omega)^{t-\tilde{j}}.$$
As proved in \cite{christ2024pseudorandom}, there exists a constant $\delta >0$ such that $p > 1/2+\delta$. 
Then, we can calculate the expected True Positive Rate (TPR), i.e., the probability of correctly identifying $m$ \ldpc-encoded vectors, as:
$$\mathsf{TPR} = \sum_{{j} \ge \tau \cdot ml} \binom{ml}{{j}} \cdot p^{{j}} \cdot (1-p)^{ml-{j}},$$
and the False Positive Rate (FPR), i.e., the probability of incorrectly identifying $m$ plain vectors as encoded, by:
$$\mathsf{FPR} = \sum_{{j} \ge \tau \cdot ml} \binom{ml}{{j}} \cdot \left(\frac{1}{2}\right)^{ml}.$$
To this point, we can choose $m$ and $\tau$ to achieve target TPR and FPR values. For example, to achieve $\mathsf{TPR}=1-\negl(\lambda)$ and $\mathsf{FPR}=\negl(\lambda)$, we can set $l, m$ such that $ml=r$ and $\tau=\frac{1}{2}+r^{-1/4}$, aligning with the $\Decode()$ function in the original~\Cref{def:pub_prc}.

\subsubsection{Complexity of Attack-I}

We then analyze the time and data complexity of Attack-I based on the MITM partial recovery and the distinguishing procedure.

\begin{thm}[Complexity of Attack-I]
\label{thm:partial_recovery}
Given a public key $\mat{G}$ and $m = \Theta(n)$ input vectors, there exists an adversary $\mathcal{A}$ that can distinguish against the undetectability property of \ldpc with advantage $| \Pr[\mathsf{Enc}] - \Pr[\mathsf{Uni}] | \ge 1 - \negl(\lambda)$ in time complexity
\begin{equation}
    \label{equ:comp_attack3_thm}
    \pr{T}_{\mypartial} = O(1) \cdot g \cdot \alpha \cdot \binom{n/2}{\lceil t/2 \rceil}
\end{equation}
and data complexity $\Theta(n)$, where $\alpha = 1 / \sqrt{qr}$ for even $t$ and $\alpha = \sqrt{\frac{\lceil t/2 \rceil}{qr(n/2-\lfloor t/2 \rfloor)}}$ for odd $t$, with $q=\sfrac{\binom{n/2}{\lceil t/2 \rceil}\binom{n/2}{\lfloor t/2 \rfloor}}{\binom{n}{t}}$.
\end{thm}
\begin{proof}
The data complexity is clearly $m = \Theta(n)$. 
For the time complexity, we can set $l = r/m = O(1)$ to achieve advantage $1 - \negl(\lambda)$ as discussed. 
Attack-I then takes time $\pr{T}_{\mypartial} + \Theta(n)$, where $\pr{T}_{\mypartial}$ is the time to recover rows of $\mat{P}$ and $\Theta(n)$ is the distinguishing complexity. Note that $\pr{T}_{\mypartial} = g \cdot \max\{|\myset{L}_1'|, |\myset{L}_2'|\}$, where the $g$ factor comes from summing rows of $\mat{G}$.
When $t$ is \textbf{even}, $|\myset{L}_1'|=|\myset{L}_2'|= \frac{1}{\sqrt{r'}} \binom{n/2}{t/2}$, so $\pr{T}_{\mypartial}^{\rm even} = \frac{g}{\sqrt{r'}} \cdot \binom{n/2}{t/2}.$
When $t$ is \textbf{odd}, $|\myset{L}_1'|=\alpha \binom{n/2}{t_1}$ and $|\myset{L}_2'|=\beta \binom{n/2}{t_2}$. To minimize MITM complexity, we set $|\myset{L}_1'|=|\myset{L}_2'|$, i.e., $\alpha \binom{n/2}{t_1}= \beta \binom{n/2}{t_2}$, with $\alpha \beta = l/r' = O(1)/r'$. Thus, $\pr{T}_{\mypartial}^{\rm odd} = g \cdot \alpha \cdot \binom{n/2}{t_1}$, where $\alpha = \sqrt{\frac{\lceil t/2 \rceil}{qr(n/2-\lfloor t/2 \rfloor)}}$.
\end{proof}





\subsection{Attack-II: Weak Key Distinguisher}
\label{sec:attackiv_weak_ket_disting}
In the following, we present an efficient \textit{multi-target} distinguishing attack that exploits an implementation vulnerability in the $\KeyGen()$ procedure rather than breaking the underlying \ldpc primitive. 
In a multi-target attack, the adversary attempts to compromise any one vulnerable key from a set of multiple public keys.
Specifically, the current PRC-based watermarking schemes (e.g., \llm~\cite{christ2024pseudorandom} and \gim~\cite{DBLP:conf/iclr/GunnZS25}) use \Cref{alg:watermark_keygeneration} or its variant to generate the random matrix $\mat{P} \overset{\$}{\leftarrow} \GF{2}^{r \times n}$, where each row of $\mat{P}$ has weight $t$. 
However, we find that rows of $\mat{G}$ can exhibit relationships that directly facilitate watermark distinguishing. 
This attack focuses on the simplest case where some rows of $\mat{G}$ are identical, and we call such $\mat{G}$ a weak key. 
The attack procedure is shown in~\Cref{fig:workflow_weak_key} and summarized in \Cref{alg:attack3_weakey}.

\subsubsection{Pipeline of Attack-II}

\mypara{Stage-I: Identifying Weak Keys}
Weak keys can be identified by detecting whether $\mat{G}$ contains identical rows, which can be accomplished in $O(ng)$ time per key using hash tables.
We now analyze the probability that a weak key is generated in \llm~\cite{christ2024pseudorandom} and \gim~\cite{DBLP:conf/iclr/GunnZS25}.

Let $\pr{P}_{\weak}^{\rm LLM}$ denote the weak key probability for the LLM scheme \llm~\cite{christ2024pseudorandom}. 
Since $\mat{G}$ is generated using~\Cref{alg:watermark_keygeneration}, recall that it first samples a uniformly random matrix $\mat{G}_0$ with $n-r$ rows, then constructs the remaining $r$ rows iteratively. 
Index these $r$ rows from $1$ to $r$. 
Ignoring the permutation ${\Pi}$, each of the last $r$ rows of $\mat{G}$ is the sum of $t-1$ rows from the first $n-r$ rows of $\mat{G}_0$. 
The number of possible distinct rows is $N = \binom{n-r}{t-1}$.
Let $p_i$ denote the probability that the first $i$ rows are distinct. 
Then $p_1 = 1$ and $p_{i+1} = p_i \cdot \frac{N-i}{N}$. 
Thus, $p_r = \frac{N(N-1)\cdots(N-r+1)}{N^r}$, and
\begin{equation}
\label{equ:pweak_llm}
    \pr{P}_{\weak}^{\rm LLM} = 1 - \frac{N(N-1)\cdots(N-r+1)}{N^r}.
\end{equation}

For $\pr{P}_{\weak}^{\rm GIM}$, the key generation implementation differs from~\Cref{alg:watermark_keygeneration} and has a larger enumeration space. 
The $i$-th row is a combination of the previous $i-1$ rows and the $n-r$ rows of $\mat{G}_0$. 
Let $p_i'$ denote the probability that the first $i$ rows are distinct. 
Then $p_1' = 1$ and $p_{i+1}' = p_i' \cdot \left(1 - \frac{i}{\binom{n-r+i}{t-1}}\right)$. 
Therefore
\begin{equation}
\label{equ:pweak_gim}
    \pr{P}_{\weak}^{\rm GIM} = 1 - \prod_{i=1}^{r} \left(1 - \frac{i-1}{\binom{n-r+i-1}{t-1}}\right).
\end{equation}

\begin{figure}[t]
\centering
\includegraphics[width=\linewidth]{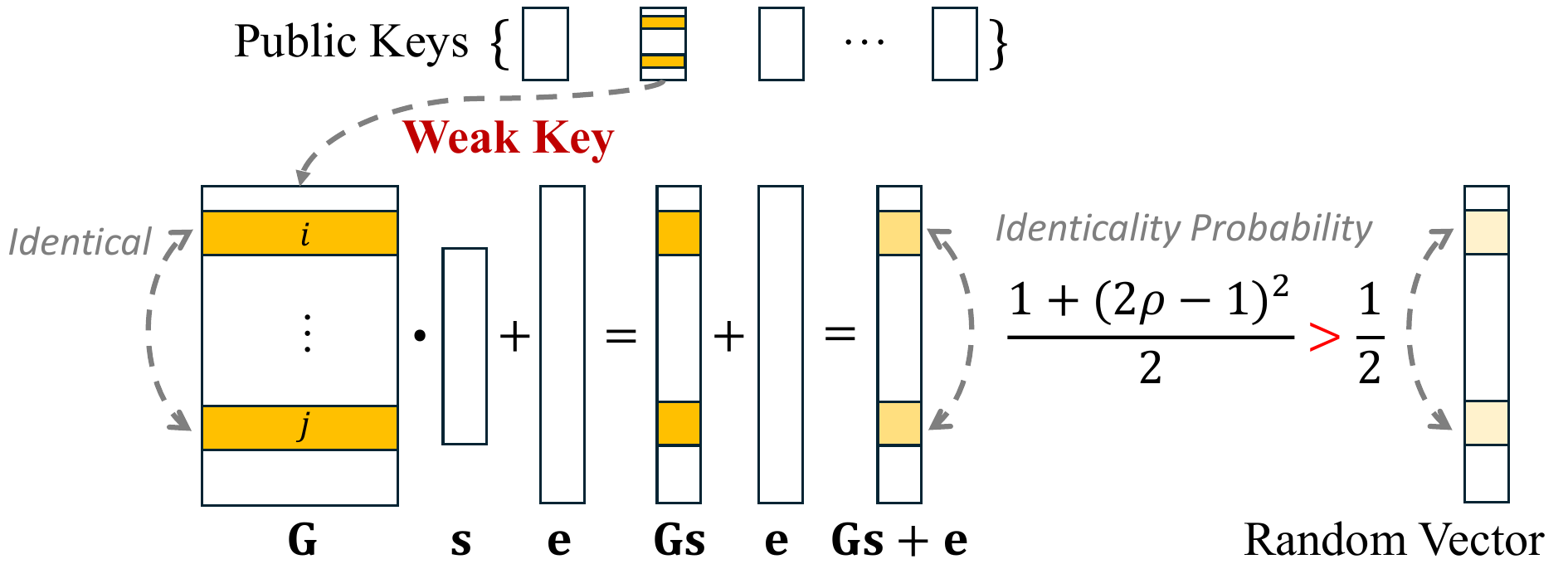}
\caption{The workflow of Weak Key Distinguisher.}
\label{fig:workflow_weak_key}
\end{figure}

\mypara{Stage-II: Distinguishing Vectors}
Assume we have $l$ pairs of row indices $\{(\alpha_j,\beta_j)\}$ such that $\mat{G}_{\alpha_j} = \mat{G}_{\beta_j}$ for all $1 \le j \le l$. 
For any $\vec{s} \in \GF{2}^{g \times 1}$, the $\alpha_j$-th and $\beta_j$-th positions of $\mat{G}\vec{s}$ are always equal. 
This holds because $\mat{G}\vec{s}$ computes a linear combination of the columns of $\mat{G}$, and identical rows ensure identical positions in the result.

Although the output $\mat{G}\vec{s} + \vec{e}$ is perturbed by the noise vector $\vec{e}$, the adversary can construct a distinguisher with constant advantage. 
Let $\vec{e}_{\alpha_j} \leftarrow \Ber(\rho)$ and $\vec{e}_{\beta_j} \leftarrow \Ber(\rho)$ with $\rho < \frac{1}{2}$. 
The probability that positions $\alpha_j$ and $\beta_j$ in $\mat{G}\vec{s}+\vec{e}$ are equal is $\rho^2 + (1-\rho)^2 = \frac{1+(2\rho-1)^2}{2}$. 
For a random vector, this probability is $\frac{1}{2}$. 
The statistical difference is $\frac{(2\rho-1)^2}{2}$, which is a constant number.

Using a threshold $\tau$ similar to~\Cref{sec:partial_sk_recovery}, we can distinguish $\mat{G}\vec{s}+\vec{e}$ from a random vector in linear time, breaking the undetectability of \ldpc. 
In practice, we can collect vectors $\{\vec{x}^{(1)},\cdots,\vec{x}^{(m)}\}$ from the oracle and calculate the ratio of vectors where positions $\alpha_j$ and $\beta_j$ are equal. 
If this ratio exceeds $\tau$, we conclude that the oracle uses \ldpc encoding. 
When $ml = r$ and $\tau = \frac{1}{2} + r^{-1/4}$, the distinguishing advantage is overwhelmingly high due to a deduction same to that in~\Cref{sec:partial_sk_recovery}.

\subsubsection{Complexity of Attack-II}

We then proceed to evaluate the time and data complexity of Attack-II based on the probability of encountering a weak key and the distinguishing procedure.

\begin{thm}[Time Complexity of Attack-II]
\label{thm:weakey_dis}
There exists an adversary $\mathcal{A}$ such that, for a \emph{fixed} key, runs in time $O(1)\cdot n \cdot g$.
If the key is weak, $\mathcal{A}$ correctly identifies it and achieves a
distinguishing advantage
$| \Pr[\mathsf{Enc}] - \Pr[\mathsf{Uni}] | \ge 1 - \negl(\lambda)
$
using $\Theta(n)$ input vectors.

Since a weak key is encountered after $\pr{P}^{-1}$ independent keys in expectation,
the total expected time complexity is
\begin{equation}
\label{equ:comp_attack4}
    \pr{T}_{\dis} = O(1) \cdot \pr{P}^{-1} \cdot n \cdot g,
\end{equation}
where $\pr{P} = \pr{P}_{\mathsf{weak}}^{\rm LLM}$ in~\Cref{equ:pweak_llm} for \llm and
$\pr{P} = \pr{P}_{\weak}^{\rm GIM}$ in~\Cref{equ:pweak_gim} for \gim.
\end{thm}
\begin{proof}
The expected cost of identifying a weak key is $O(1) \cdot \pr{P}^{-1} \cdot n \cdot g$, where $O(n \cdot g)$ is the cost of detecting duplicate rows in a single key, and the $\pr{P}^{-1}$ is the number of expected keys. 
Given a weak key $\mat{G}$, the adversary can distinguish $\mat{G}\vec{s}+\vec{e}$ from a random vector in linear time. 
A distinguishing advantage of $1 - \negl(\lambda)$ is achieved by processing $\Theta(n)$ input vectors. 
Therefore, the total expected time complexity is $\pr{T}_{\dis} = O(1) \cdot \pr{P}^{-1} \cdot n \cdot g + \Theta(n)$. 
The $\Theta(n)$ term is omitted from the final complexity as it is dominated by the first term.
\end{proof}


\subsection{Attack-III: Noise Overlay Attack}
\label{subsec:msg_recovery}
For PRCs, the noise vector is modeled as a Bernoulli noise channel as defined in~\Cref{def:prc}. 
This noise can originate from multiple sources, such as the inherent noise in LLM or GIM applications. 
For instance, in LLM watermarking, noise arises during the sampling of $t_i$ from $(p_i,x_i)$ by~\Cref{alg:TokenSample}, while in GIM, noise is introduced by the noise inversion procedure (i.e., $z(\cdot)$ of~\Cref{alg:GIMInv}). 
However, we highlight that an adversary may also actively modify content to evade detection. 
To investigate this threat, we propose a noise overlay attack that challenges the robustness of PRCs. 
The attack procedure is shown in~\Cref{fig:workflow_attack_1} and summarized in \Cref{alg:attack4_erroroverlay}.

\begin{figure}[b]
\centering
\includegraphics[width=0.9\linewidth]{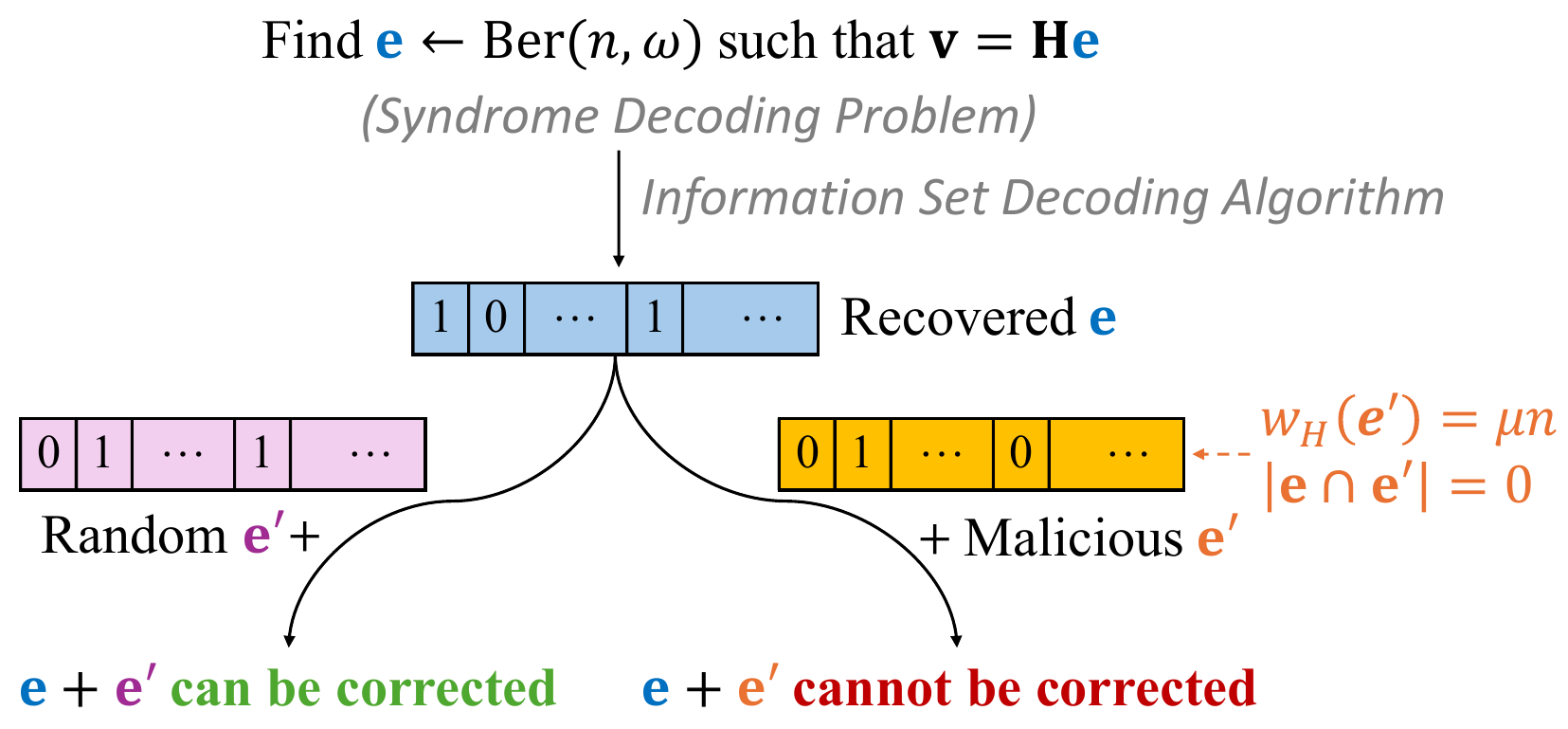}
\caption{The workflow of Noise Overlay Attack.}
\label{fig:workflow_attack_1}
\end{figure}

\subsubsection{Pipeline of Attack-III}

Given an encoded vector $\vec{x}=\mat{G}\vec{s}+\vec{e}+\vec{z}$ and public key $\mathsf{pk}=(\mat{G},\vec{z})$, the goal of Attack-III is to first recover the noise vector $\vec{e}$, then construct a malicious noise vector $\vec{e}'$ such that $\vec{x}+\vec{e}'$ cannot be successfully decoded. 
Since \cite{christ2024pseudorandom} does not specify an explicit noise rate bound, our objective is to maximize the resulting noise rate by adding $\vec{e}'$ with a \textit{fixed} noise rate.
We assume the decoding process fails once the final noise rate exceeds the designed threshold. 
Note that such a malicious noise attack is non-trivial because, under the same fixed noise rate, a PRC-based codeword with additional random noise of identical weight can be successfully decoded with high probability, while our malicious noise vector causes decoding to fail. 
We now detail the two steps of the attack.

\begin{itemize}
\item \textbf{Step I: Noise Recovery.}
Upon receiving $\vec{x}$, we first compute $\vec{x}'=\vec{x}+\vec{z}=\mat{G}\vec{s}+\vec{e}$. 
Let $\mat{H} \in \GF{2}^{(n-g) \times n}$ be the dual parity-check matrix of $\mat{G}\in \mathbb{F}^{n\times g}_2$ satisfying $\mat{HG}=\mat{0}$. 
Next, we can get $\vec{v} = \mat{Hx'}=\mat{HGs}+\mat{He}=\mat{He}$. 
At this stage, we need to find a vector $\vec{e}\leftarrow \Ber(n,\omega)$ satisfying $\mat{He}=\vec{v}$.
This reduces to the classical Syndrome Decoding (SD) problem~\cite{berlekamp2003inherent}, which is computationally equivalent with the LPN problem in \ldpc. 
Consequently, it can be solved using ISD algorithms~\cite{prange1962use,dumer1991minimum,stern1988method,may2011decoding,becker2012decoding,may2015computing}. 
Although the exact weight of $\vec{e}$ is unknown, we consider the upper bound of error-correcting capability where the noise rate equals $\frac{1}{2}-\epsilon$, as calculated by the revised inequality in~\Cref{lem:t_choice}.

\item \textbf{Step II: Noise Overlay.}
Recall that the received vector is $\vec{x}'=\mat{G}\vec{s}+\vec{e}$ where $\vec{e} \leftarrow \Ber(n,\omega)$. 
If we add another random $\vec{e'} \leftarrow \Ber(n,\mu)$ to $\vec{x}'$, where $\vec{e}$ and $\vec{e}'$ are independent, then $\Pr[e_i+e'_i=1] = \omega+\mu-2\omega\mu$ and $\mathbb{E}[\wt{\vec{e}+\vec{e'}}]=(\omega+\mu-2\omega\mu)n$. 
However, once $\vec{e}$ is recovered, we can maliciously construct a special $\vec{e}'$ satisfying both $\wt{\vec{e'}}=\mu n$ and $|\vec{e} \cap \vec{e}'|=0$, such that
$$\mathbb{E}[\wt{\vec{e}+\vec{e'}}]=(\omega+\mu)n \gg (\omega+\mu-2\omega\mu)n.$$
Assume \ldpc is designed to tolerate additional random noise with rate $\mu_0$, which corresponds to an overall noise rate of $(\omega + \mu_0 - 2\omega\mu_0)$.
Hence, we can set $\mu = \mu_0$ in our attack.
While the resulting randomly combined noise remains detectable, the adversarially constructed noise vector $\vec{e} + \vec{e}'$ has weight $(\omega + \mu_0)n$, exceeding the decoding capability of the \ldpc.
Therefore, introducing $\vec{e}'$ successfully causes decoding failures.
\end{itemize}

\subsubsection{Complexity of Attack-III}
In this part, we derive the time and data complexity of Attack-III in Theorem~\ref{thm:error_recover}.
In particular, we adopt the Prange algorithm~\cite{prange1962use} to estimate the time complexity.

\begin{thm}[Complexity of Attack-III]
\label{thm:error_recover}
Given a public key $\mathsf{pk}$ and an encoded vector $\mat{x}$, there exists an adversary that can construct $\vec{x}''=\echannel(\vec{x})$ to achieve $\mathsf{Decode}(1^\lambda,\mathsf{sk},\vec{x}'') \neq \mathsf{m}$ with time complexity
\begin{equation}
\label{equ:comp_attack1}
    \pr{T}_{\ISD} = \left(\frac{1}{2}+\epsilon\right)^{-g} \cdot n^{3}.
\end{equation}
\end{thm}

\begin{proof}
For an SD problem $\mat{H}\vec{e}=\vec{v}$, we first randomly shuffle the columns of $\mat{H}$. 
We then calculate the probability that this shuffle concentrates all nonzero elements of $\vec{e}$ in the first $n-g$ positions, leaving the last $g$ positions as zeros (i.e., $\wt{\vec{e}_{n-g+1:n}}=0$). 
Since each element of $\vec{e}$ equals $1$ with probability $\frac{1}{2}-\epsilon$, this probability is
$$\pr{P}_{\mathsf{nz}}=\left(\frac{1}{2}+\epsilon\right)^g.$$
If this condition holds, the SD problem can be solved via Gaussian elimination with $n^3$ operations. 
Thus, the total time complexity for the Prange algorithm is $\pr{T}_{\mathsf{Prange}}=\pr{P}_{\mathsf{nz}}^{-1} n^3$. 
Since the remaining steps of adding $\vec{e}'$ require only $O(n) \ll \pr{T}_{\mathsf{Prange}}$ time, the total complexity is dominated by noise recovery
$$\pr{T}_{\ISD} = \left(\frac{1}{2}+\epsilon\right)^{-g} n^{3}.$$
\end{proof}

\mypara{Remark} In \gim, we note that the weight of added noise $\eta$ is much lower than the theoretical bound $\frac{1}{2}-\epsilon$. 
Let $\pr{T}'_{\ISD}$ denote the cost of Attack-III with the concrete noise rate $\eta$ instead of $\frac{1}{2} - \epsilon$.
Consequently, we can estimate $\pr{T}'_{\ISD}$ as $$\pr{T}'_{\ISD}=(1-\eta)^{-g}\cdot n^3.$$

\section{Concrete Time Complexity Analysis}
\label{sect/concrete_complexity}
Building on our theoretical analysis, we now analyze the complexity of our attacks under concrete parameter choices of watermarking schemes.
Although the attacks initially operate on watermarked images or texts, we first transform them into the corresponding codewords via~\Cref{alg:TextInv} and~\Cref{alg:GIMInv}.
As the transformation cost is negligible compared to the attacks themselves, we focus solely on the complexity of our three attacks.

\subsection{Parameter Configurations}

We summarize the parameter configurations of \ldpc for \llm and \gim in~\Cref{tab:concrete_paras}, where the first two rows present the default settings provided by the original \llm and \gim.
However, we notice that some critical parameter settings remain unspecified, e.g., $n$ and $\lambda$ in \llm.
Therefore, we determine two parameter configurations, Config~\ref{config_llm} and Config~\ref{config_gim}, for \llm and \gim, respectively.
These specific parameters facilitate a concrete time complexity analysis of the proposed attacks.


\subsubsection{Parameter Configurations for \llm}
\label{subsec:para_setup_llm}



\begin{table}[t]
\centering
\caption{ Parameters of \ldpc for \llm and \gim.
}
\setlength{\tabcolsep}{5pt}
\begin{tabular}{c|c|c|c|c|c}
\toprule
  &  $n$  & $r$ & $g$ & $\lambda$ & $t$  \\  
\midrule 
\llm~\cite{christ2024pseudorandom} & - & $0.99n$ & $c\log^2n$ & $-$ & $\log n$ \\
\midrule
\gim~\cite{DBLP:conf/iclr/GunnZS25} & $2^{14}$ & $n-k-\lambda$ & $\log \binom{n}{t}$ & $g$ & $3 \le t \le 7$ \\ \midrule  
\makecell[c]{\textbf{Config 1} \\ (for \llm)} & $2^{17}$ & $0.99n$ & $\log \binom{n}{t}$ & $g$ & $3 \le t \le 14$\\
\midrule
\makecell[c]{\textbf{Config 2} \\ (for \gim)} & $2^{14}$ & $n-k-\lambda$ & $\log \binom{n}{t}$ & $g$ & $3 \le t \le 7$ \\
\bottomrule 
\end{tabular}
\label{tab:concrete_paras}
\end{table}

\begin{configuration}
\label{config_llm}
For \llm, we set $n \le 2^{17}$, $r\le 0.99n$, $g=\lambda=\log \binom{n}{t}$, and $3 \le t \le 14$.
\end{configuration}
\noindent When analyzing the computational complexity of attacks against \llm, we use the maximum parameter values in Config~\ref{config_llm} to provide the upper bound of security: $n=2^{17},r=\lfloor 0.99n \rfloor$, and $g=\lambda=\log \binom{n}{t}$ for each $t$.
Next, we explain the rationales for the parameter settings in Config~\ref{config_llm} as follows:
\begin{itemize}
\item \textbf{The code length $n$}: According to \Cref{construction_llm}, each token is mapped into binary bit strings.
We assume the length of binary bit strings is smaller than $32$ bits, i.e., the size of the token alphabet\footnote{The size of the alphabet varies from thousands to millions; thus, $2^{32}$ is a very loose upper bound.} is smaller than $2^{32}$.
Meanwhile, we assume that the max output length is $2^{12}=4,096$ tokens, which is a common setting for LLMs (e.g.,  gpt-3.5-turbo~\cite{gpt35url} and gpt-4-turbo~\cite{gpt4url}).\footnote{We also notice that the recently released, more advanced LLMs are capable of supporting $32,768$ max output tokens, e.g., gpt-4.1-turbo~\cite{gpt41url}.
We will prove that \llm can still not reach an $128$-bit security even with such length.
More details are available in~\Cref{subsec:defense_llm}.}
Note that in practical scenarios, the output length of LLMs typically does not reach the maximum limit. 
If the defender requires stronger robustness against the binary deletion channel (BDC) as described in \cite{christ2024pseudorandom}, the \ldpc must be embedded into a longer code. Consequently, the length of \ldpc becomes significantly shorter than the output length of LLMs.
As a result, it is natural and reasonable to assume that the code length $n$ is not larger than $32 \times 2^{12}=2^{17}$.


\item \textbf{The column size} $g$ \textbf{and security parameter} $\lambda$: Though  \cite{christ2024pseudorandom} sets $g=c\log^2 n$ where $c$ is a constant, but no concrete value of $c$ and expression of $\lambda$ is given.
Hence, we follow \gim~\cite{DBLP:conf/iclr/GunnZS25} to set $g=\lambda=\log \binom{n}{t}$.

\item \textbf{The Row-wise nonzero count} $t$: 
We set $3 \le t \le 14$ for \llm due to the following reasons.
(i) The lower bound of $t$: When $t$ is $1$ or $2$, the complexity of  Attack-I is linear in $n$.
Therefore, to exclude this trivial attack, we restrict our analysis to $t \ge 3$.
(ii) The upper bound of $t$: We use~\Cref{lem:t_choice} to calculate the upper bound of $t$.
Let $\rho = 1/2-\epsilon$ be the maximum error-correcting capability.
Since no explicit value of $\rho$ is discussed in \cite{christ2024pseudorandom}, we use the experimental results in \gim as reference.
During the experiments in~\Cref{subsec:realworld_gim}, we find that \gim has to correct vectors with noise rate $\rho=0.074$.
Hence, we bound $\rho \ge 0.074$, i.e., $\epsilon \le 0.426$.
However, when we try to calculate the upper bound of $t$ with $r$ and $\epsilon$ following~\cite{christ2024pseudorandom}, we notice a mistake in their calculation, i.e.,
$$(2\epsilon)^t / 2 > r^{-1/4} \nLeftrightarrow t<1+\log(r)/4\log(1/2\epsilon).$$
Therefore, we give a corrected version in~\Cref{lem:t_choice} and get $t \le 14$.



\end{itemize}




\begin{lemma} \label{lem:t_choice}
In order to decode successfully with high probability, we need $(2\epsilon)^t / 2 > r^{-1/4}$, i.e., $$t < \frac{\frac{1}{4}\log r-1}{\log ({1}/{2\epsilon})}.$$
\begin{proof}
With $(2\epsilon)^t / 2 > r^{-1/4}$,
we have $ (2\epsilon)^t > 2r^{-1/4} $.
Since $0 < \epsilon < 1/2$, we can simplify the inequality to
$ t < \log_{2\epsilon}(2r^{-1/4})$, i.e.,
$ t < \frac{1-1/4\log{r}}{\log (2\epsilon)} = \frac{\frac{1}{4}\log{r}-1}{\log{(1/2\epsilon)}}$.
\end{proof}
\end{lemma}

\subsubsection{Parameter Configurations for \gim}
\label{subsec:para_setup_gim}

\begin{configuration}
\label{config_gim}
For \gim, we set $n=2^{14}$, $g=\lambda=\log \binom{n}{t}$, $k=\lambda + {message\_bit} + {parity\_bit}$
where $message\_bit=512$, and $parity\_bit= \log_2({\rm FPR})$ where ${\rm FPR}=10^{-5}$.
When calculating $\vec{y}=\mat{G}\vec{s}+\vec{e}$, the added noise vector $\vec{e}$ has noise weight $\eta=1-2^{{-\lambda}/{g^2}}$.
\end{configuration}

For Config~\ref{config_gim}, we employ the default parameters from~\cite{DBLP:conf/iclr/GunnZS25}.
Note that the weight $t$ is set to $3$ in most cases discussed in~\cite{DBLP:conf/iclr/GunnZS25}, and it is also suggested to set $t=\log(n)/2=7$ to ensure cryptographic undetectability.
For a complete analysis of our attacks, we vary the weight $t$ from $3$ to $7$ and change $r,g,k,\lambda$, and $\eta$ accordingly.





\subsection{Overall Performance}
\label{subsec:overall_performance}

First of all, we present a comprehensive comparison between the complexity of our attacks and the claimed security levels of \ldpc, and illustrate the trend of complexity with varying parameters in~\Cref{fig:all_complexities}.
The detailed complexities are available in Table~\ref{tab:complexity_llms} and Table~\ref{tab:complexity_gims} (Appendix~\ref{supp/attack_complexity}).

\mypara{Performance on \llm}
First, as shown in \Cref{fig:llm_complexities}, the best complexities of the three attacks are lower than the claimed security level $\lambda$ across all parameters. 
Second, it presents clear trends in the complexities.
For Attack-I and Attack-II, their respective attack complexities (i.e., $\pr{T}_{\mypartial}$, and $\pr{T}_{\dis}$) exhibit a consistent increasing trend with growing $t$.
Nevertheless, it consistently remains below $128$ bits, demonstrating that our attacks remain feasibly threatening.
Notably, as shown in Table~\ref{tab:complexity_llms}, for Attack-II, the frequency of weak keys is non-negligible and approaches $100\%$ for parameters $t=3$ and $t=4$, making the attack always effective in these cases.
For Attack-III, we can observe that even as $t$ increases, the value of $\pr{T}_{\ISD}$ remains nearly unchanged and consistently stays around $73$ bits, indicating that \llm exhibits vulnerability to our noise overlay attack across a reasonable range of parameters.


\begin{figure}[t]
\centering
\begin{subfigure}[b]{0.23\textwidth}
\includegraphics[width=\linewidth]{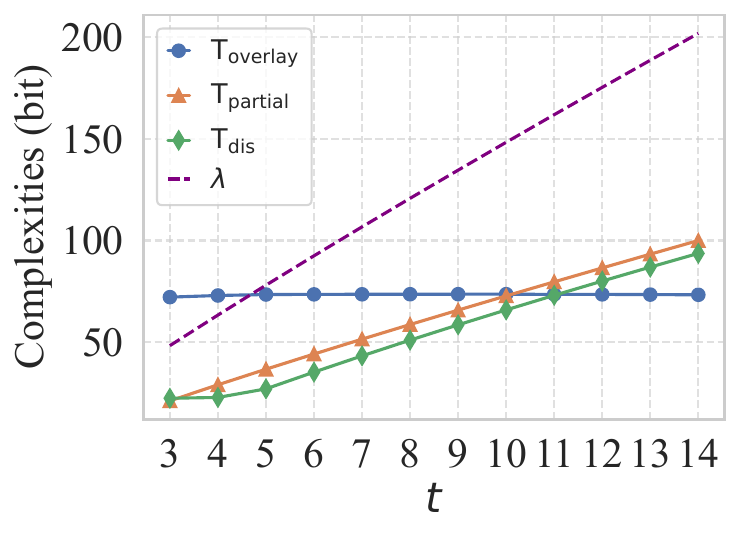} 
\caption{LLM.}
\label{fig:llm_complexities}
\end{subfigure}
\hfill 
\begin{subfigure}[b]{0.23\textwidth}
\includegraphics[width=\linewidth]{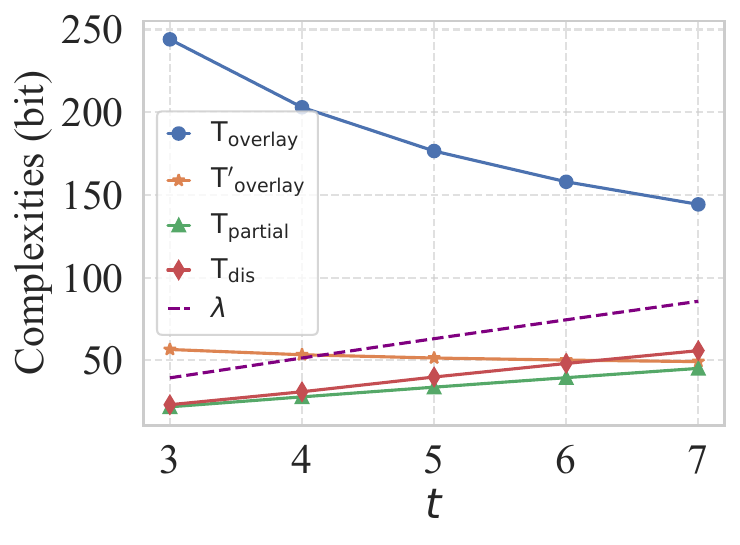} 
\caption{GIM.}
\label{fig:img_complexities}
\end{subfigure}
\caption{Time complexities of all attacks. The complexities are expressed in logarithm form.}
\label{fig:all_complexities}
\end{figure}

\mypara{Performance on \gim}
Similarly, \Cref{fig:img_complexities} shows that the best complexities of our three attacks remain below the claimed security level~$\lambda$ for all parameter settings. 
The first two attacks exhibit increasing complexity with growing $t$, which is similar to the performance on \llm.
Although $\pr{P}_{\weak}^{\rm GIM}$ of GIMs is larger than that of LLMs, it remains high under concrete parameters.
In Table~\ref{tab:complexity_gims}, when $t=3$, weak keys occur with nearly $100\%$, indicating practical vulnerability.
For Attack-III, we can observe that both the complexity $\pr{T}_{\ISD}$ and $\pr{T}_{\ISD}'$ decrease as $t$ increases, which comes from the decrease in the theoretical noise rate $\rho$ and the concrete noise rate $\eta$.
Due to the gap between $\rho$ and $\eta$ , there is a huge difference between $\pr{T}_{\ISD}$ and $\pr{T}_{\ISD}'$.
Taking into account $\pr{T}_{\ISD}'$ instead of $\pr{T}_{\ISD}$, all complexities are below $64$ bits, indicating practical threats under the parameter choices of \gim.

\subsection{Further Analysis of Attack-III}
\label{subsec:concrete_complexity_attack1}

\mypara{Performance on \llm}
The experimental results show that $\pr{T}_{\ISD}$ increases slightly as $t$ increases when $3 \le t \le 10$, but decreases when $ t \ge 12$.
This phenomenon indicates that a larger $t$ does not always lead to a higher security level.
We also estimate all reasonable parameters on $t>14$, which concretely verifies this conclusion.
We further analyze this phenomenon on the basis of the following truth.
\begin{itemize}
\item $\pr{T}_{\ISD} \propto (1-\rho)^{-g }=(\frac{1}{2}+\epsilon)^{-g}$.
\item $g$ increases as $t$ increases since $g=\log \binom{n}{t}$.
\item $\epsilon$ increases as $t$ increases because of the constraint in~\Cref{lem:t_choice}.
\end{itemize}
Therefore, for small values of $t$, the term $g$ dominates the behavior of $\pr{T}_{\ISD}$, causing $\pr{T}_{\ISD}$ to increase as $t$ increases.
As $t$ progressively increases, $\epsilon$ gradually becomes dominant, leading to a decrease in $\pr{T}_{\ISD}$.
Second, although the estimated $\pr{T}_{\ISD}$ seems large for small $t$ compared to the other attacks, the defender might add a much lower noise rate as \llm does, thereby causing a much lower $\pr{T}_{\ISD}$. 


\mypara{Performance on \gim}
$\pr{T}_{\ISD}$ decreases sharply when $3 \le t \le 7$ since the noise weight dominates the complexity.
However, the actual noise weight $\eta \ll \rho$ in all cases, making $\pr{T}_{\ISD}'$ much lower than the theoretical bound $\pr{T}_{\ISD}$.
The gap could be as large as $187$ bits, indicating the practical threat of our attack.

\section{Real-World Evaluations}
\label{sect/real_world_attack}
To further demonstrate the practical threat posed by our attacks, we launch Attack-I and Attack-II against real-world large generative models.
For Attack-III, the text quality of the watermarked \llm is low, making it difficult to analyze the effect of our modification. 
Therefore, we conduct Attack-III only on \gim.
Moreover, to verify that our attacks can generalize beyond specific experimental settings, we also conduct supplementary experiments on different models in Appendix~\ref{supp/diff_watermarks}.

\subsection{Setups}
\label{subsec:setups}
\mypara{Target LLM}
We employ DeepSeek-R1-Distill-Qwen-7B~\cite{deepseekai2025deepseekr1incentivizingreasoningcapability} as our real-world target LLM.
Each response is generated with a fixed length of $1,024$ tokens, chosen from a vocabulary of size $152{,}064$.
To embed watermarks, we apply \llm~\cite{christ2024pseudorandom} as described in~\Cref{sec:llm-scheme}.
Following Config~\ref{config_llm}, we set $n=\lceil \log_2 (152,064) \rceil * 1,024=18,432$, $r=\lfloor 0.95n \rfloor =17,510$, and $t$ as $3$ or $4$.


\mypara{Target GIM}
We adopt stable-diffusion-2-1-base~\cite{Rombach_2022_CVPR} as our target GIM.
To generate images, we use the Multistep DPMSolver Scheduler~\cite{10.5555/3600270.3600688} as the de-noising scheduler $f(\cdot)$ of \Cref{alg:GIM} and set the de-noising step $s$ to $50$.
We apply the \gim~scheme to embed watermarks with $t=$ $3$ or $4$, and all other parameters are set as in Config~\ref{config_gim}.
During watermark verification, we perform DDIM Inversion~\cite{song2020denoising} as the noise inversion scheduler $z(\cdot)$ of~\Cref{alg:GIMInv} with the same number of steps and guidance scale as the image generation process.

\mypara{AI-Generated Contents}
In each set of parameters, we randomly generate $128$ pairs of $\mathsf{pk}=(\mat{G},\vec{z})$. 
For each $\mathsf{pk}$, we generate $16$ texts or images. 
Meanwhile, we also generate non-watermarked content with the same number as a control group.
For GIMs, we follow the model hyperparameter settings in \cite{DBLP:conf/iclr/GunnZS25} and sample the prompts from the test split of the Gustavosta/Stable-Diffusion-Prompts dataset~\cite{diffusion-prompts}.

\mypara{Settings of Attack-I and Attack-II}
When attacking \llm, we set the temperature to $1.8$, and the parameters for attacking are set as follows: for $t=3$, we set $|\myset{L}'_1|=|\myset{L}'_2|=18,432$ and $\tau=0.60$ for both Attack-I and Attack-II.
When $t=4$, we set $|\myset{L}'_1|=|\myset{L}'_2|=5 \times 10^6$, and $\tau = 0.55$ for Attack-I and $\tau = 0.65$ for Attack-II.
The parameter settings for attacking \gim are as follows:
When $t=3$, we set $|\myset{L}'_1|=|\myset{L}'_2|=16,384$ and $\tau=0.75$ for both Attack-I and Attack-II.
When $t=4$, we set $|\myset{L}'_1|=|\myset{L}'_2|=5 \times 10^5$ and $\tau=0.75$ for both Attack-I and Attack-II.

\mypara{Evaluation Metrics}
(1) For Attack-I, we define the success rate as the probability that $|\myset{L}'| > 0$.
For Attack-II, the success rate corresponds to the probability of encountering weak keys.
(2) We also evaluate the TPR and FPR for watermark detection upon the successful attacks.
Let $\text{TPR}_0$ denote the TPR when testing the original codeword generated by the defender, i.e., $\vec{x}=\mat{G}\vec{s}+\vec{e}+\vec{z}$, and $\text{FPR}_0$ denote the FPR when testing a uniformly random vector.
Meanwhile, we use $\text{TPR}_1$ to denote the TPR for the extracted codeword from the watermarked content, which includes both the defender’s added noise and the noise introduced by the schemes.
Specifically, the additional noise of \llm is introduced in Line 5 and Line 7 of~\Cref{alg:TokenSample}, while the noise of \gim arises from the error of DDIM inversion.
Correspondingly, $\text{FPR}_1$ represents the FPR for the codeword corresponding to the unwatermarked content.



\subsection{Performance on LLMs}
\label{subsec:realworld_llm}


\begin{figure}[t]
\centering
\includegraphics[width=\linewidth]{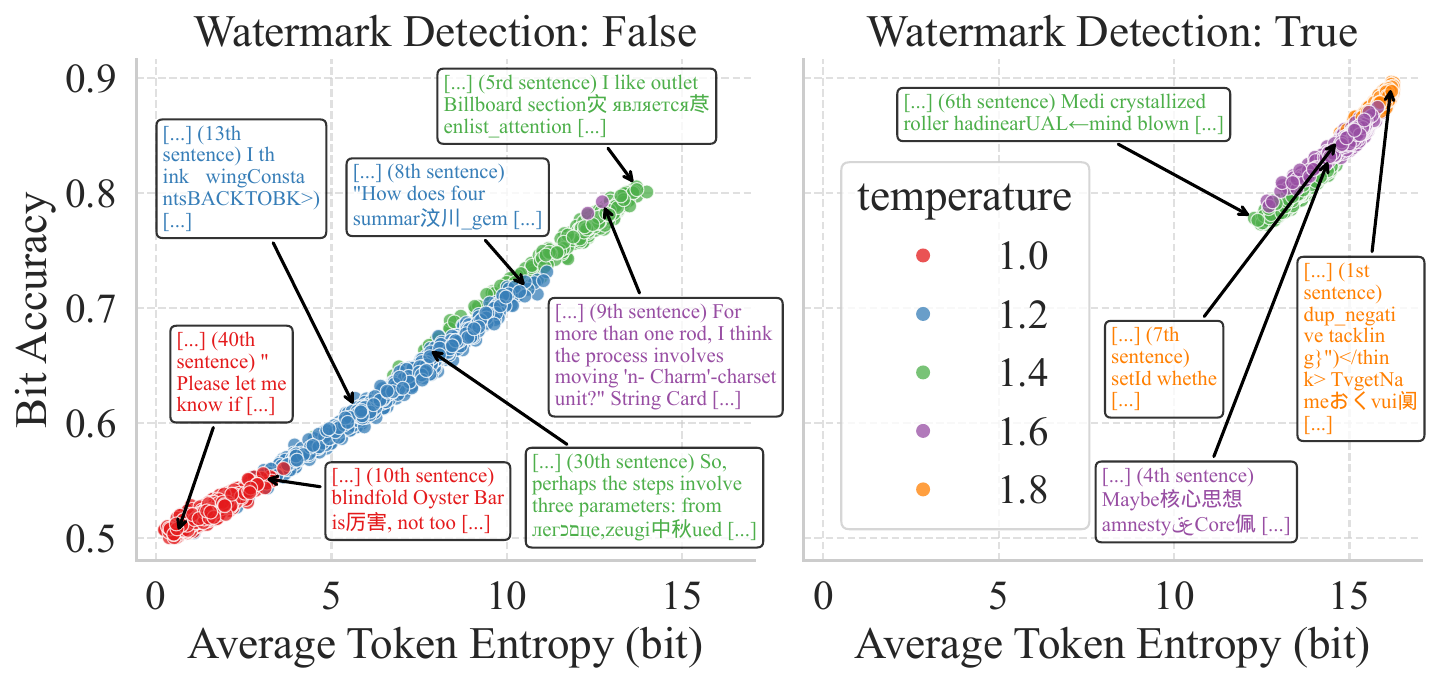}
\caption{The relationship between the readability of generated text and the detectability of watermarks.
We set $t=3$ for \llm.}
\label{fig:avg_entropy_vs_correct_rate}
\end{figure}

\subsubsection{Impracticality of \llm}

Prior to an in-depth examination of the attack effectiveness, it is imperative to first elucidate the impracticality of \llm.
Recall in Algorithm~\ref{alg:TokenSample}, the tokens are biased according to the probability $p'$ derived from the codewords, so the bit accuracy of the recovered message depends on the average token entropy of the LLM output.
To satisfy \ldpc’s correction bounds, the LLM must produce high-entropy token distributions, which correspond to low-confidence predictions and thus low-quality responses.

To examine this, we generate watermarked texts using temperature values $\mathsf{temp}$ from $1.0$ to $1.8$, controlling output entropy.
As shown in~\Cref{fig:avg_entropy_vs_correct_rate}, the watermark is undetectable at temperatures of $1.0$ and $1.2$.
At $\mathsf{temp}=1.4$, watermarks are detectable in 60\% of outputs, but readability is severely degraded, with frequent language mixing (green texts in~\Cref{fig:avg_entropy_vs_correct_rate}).
At $\mathsf{temp}=1.6$ and $1.8$, detection rates further increase, but the generated text becomes nearly unreadable (purple and orange texts in~\Cref{fig:avg_entropy_vs_correct_rate}).
Therefore, we launch attacks against texts capable of carrying watermarks, disregarding their readability.




\subsubsection{Performance of Attack-I}

As described in Table~\ref{tab:concrete_experiment_llm}, Attack-I maintains a perfect success rate of 100\%, which means we can stably construct nonempty $\myset{L}'$. 
More precisely, we can recover an average of 8.8 rows of the secret keys.
Meanwhile, the attainment of 100\% $\text{TPR}_0$ alongside 0\% $\text{FPR}_0$ confirms that Attack-I possesses a requisite precision to discriminate between PRC codewords and random vectors.
However, the $\text{TPR}_1$ on the real-world output texts becomes 61\% when $t=3$ and 98\% when $t=4$.
This degradation is attributable to the additional noise introduced during the token sampling procedure in~\Cref{alg:TextInv}. 
Specifically, the error rate increases from $0.10$ for the original codeword to $0.20$ for the recovered codeword.

\subsubsection{Performance of Attack-II}
On the other hand, Attack-II achieves complete success at $t=3$, yet drops to 63\% when $t=4$.
It demonstrates a considerable likelihood of weak keys occurring within LLM scenarios.
Furthermore, the high $\text{TPR}_0$ values and low $\text{FPR}_0$ values collectively demonstrate Attack-II's efficacy in differentiating between PRC codewords and random vectors, demonstrating that with up to $22.09$ duplicated rows on average, the distinguisher performs very well.
Unlike Attack-I, the $\text{TPR}_1$ sustains a perfect 100\% on real-world texts at $t=3$.
When $t=4$, the gap between $\text{TPR}_1$ and $\text{FPR}_1$ remains sufficiently large to enable reliable detection.



\begin{table}[t]
\centering
\caption{Watermark detection attacks against real-world LLM. 
}
\setlength{\tabcolsep}{5pt}
\begin{tabular}{c|c|c|c|c|c|c}
\toprule
Attack & $t$ & $\text{TPR}_0$ & $\text{FPR}_0$ & $\text{TPR}_1$ & $\text{FPR}_1$ & Success Rate  \\ 
\midrule
\multirow{2}{*}{I} & 3  & $100\%$ & $0\%$ & $61\%$ & $0\%$ & $100\%$ \\ 
 & 4  & $100\%$ & $0\%$ & $98\%$ & $0\%$  & $100\%$ \\ 
\midrule
\multirow{2}{*}{II} & 3 & $100\%$ & $0\%$ & $100\%$  &  $0\%$ &$100\%$ \\ 
& 4 & $98\%$ & $10\%$ & $59\%$ & $10\%$  & $63\%$ \\ 
\bottomrule
\end{tabular}
\label{tab:concrete_experiment_llm}
\end{table}

\subsection{Performance on GIMs}
\label{subsec:realworld_gim}

\subsubsection{Attack Scenario}

Given an image, both defenders verifying watermarks and adversaries launching detection attacks must utilize Algorithm~\ref{alg:GIMInv} to recover the codeword.
Algorithm~\ref{alg:GIMInv} employs two neural networks, \texttt{U-net} and \texttt{E-net}.
The defender typically employs the same \texttt{U-net} utilized in the generation process, along with an \texttt{E-net} compatible with the \texttt{D-net} that was used for generating the image, for a highly precise codeword extraction.
For the adversary, we first assume that they can obtain the same \texttt{U-net} and \texttt{E-net} as the defender.
However, considering that model parameters are typically kept confidential as commercial secrets, we also consider that the adversary will employ a proxy model that differs from the defender's.

\subsubsection{Performance of Attacks}

\mypara{Distinguishing Attacks}
The performance of distinguishing attacks against \gim is shown in Table~\ref{tab:concrete_experiment_gim}.
In our experiments, we find the error rate of the original codeword is $0.018$ and the noise rate of the received codeword is $0.074$.
With a lower-weight noise vector compared to \llm, we achieve more effective attacks.
When $t=3$, although only $5.8$ rows are recovered on average (vs. $8.8$ in \llm) in Attack-I, the $\text{TPR}_1$ is $93\%$ (vs. $61\%$ in \llm) with FPR remaining $0\%$.
For Attack-II, the weak key frequency is $100\%$ and duplicated rows average $14.86$.
Our distinguisher also performs well with $\text{TPR}_1=96\%$ and $\text{FPR}_1=0\%$.
The cases of $t=4$ are similar to $t=3$.
Meanwhile, to verify that our attacks also work without access to the target model's neural network in real-world scenarios, we conduct extra experiments with other proxy models, and the results show that our attacks still succeed.
Detailed results under different proxy models are discussed in Appendix~\ref{supp/diff_ddim_inv}.

\mypara{Removal Attacks} 
Then, we investigate leveraging Attack-III to remove the watermark embedded by \gim.
Given Attack-III, the attacker can recover $\vec{x}-\vec{e}$.
Subsequently, Projected Gradient Descent (PGD)~\cite{madry2018towards} is applied to search for an adversarial perturbation $\mathsf{Err}$ within $\delta$-bounded $l_\infty$ norm, such that the sign of the reversed codeword aligns with $-(\vec{x}-\vec{e})$.
Such an optimization objective is formalized as
$$\arg\min_{\mathsf{Err}} \; 
\left\| \tanh [k \cdot \texttt{ImageToCodeword}(\mathsf{Img} + \mathsf{Err}) ] + (\vec{x}-\vec{e}) \right\|_2,$$
where $k$ is the hyperparameter setting to $5$.
We conduct removal attacks on 100 images watermarked with \gim under $t=3$.
After 100 optimization steps with a learning rate of 0.01, the watermark removal success rate is $67\%$ under $\delta=8/255$, $89\%$ under $\delta=12/255$, and $96\%$ under $\delta=16/255$, with the visual examples shown in~\Cref{fig:removal}.
These results demonstrate that Attack-III can be effectively leveraged to remove watermarks against \gim.

\begin{table}[t]
\centering
\caption{Watermark detection attacks against real-world GIM. 
}
\setlength{\tabcolsep}{5pt}
\begin{tabular}{c|c|c|c|c|c|c}
\toprule
Attack & $t$ & $\text{TPR}_0$ & $\text{FPR}_0$ & $\text{TPR}_1$ & $\text{FPR}_1$ & Success Rate  \\ 
\midrule
\multirow{2}{*}{I} & 3  & $99\%$ & $0\%$ & $93\%$ & $0\%$ & $100\%$ \\ 
 & 4  & $93\%$ & $4\%$ & $66\%$ & $4\%$  & $78.9\%$ \\ 
\midrule
\multirow{2}{*}{II} & 3 & $100\%$ & $0\%$ & $96\%$  &  $0\%$ &$100\%$ \\ 
& 4 & $100\%$ & $0\%$ & $100\%$ & $0\%$  & $2.3\%$ \\ 
\bottomrule
\end{tabular}
\label{tab:concrete_experiment_gim}
\end{table}

\begin{figure}[t]
    \centering
    \includegraphics[width=\linewidth]{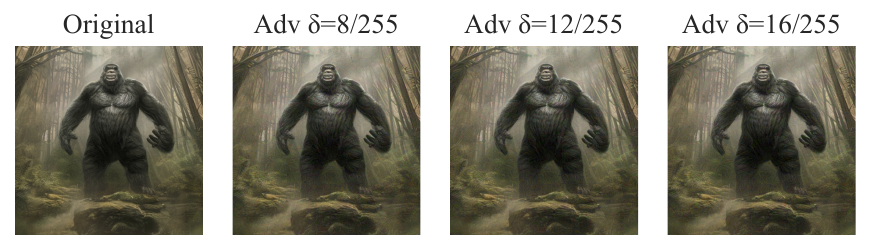}
    \caption{Example images with removed watermark.}
    \label{fig:removal}
\end{figure}

\section{Mitigations}
\label{sect/defense}

\vspace{2mm}

\subsection{Parameter Suggestions}
\label{subsec:defense_llm}
We show how Attack-I and Attack-III can be mitigated through appropriate parameter selection in Theorem~\ref{thm:mitigation_complexity}.

\begin{thm}[Complexities under Different Parameters]
\label{thm:mitigation_complexity}
Following the parameter relationships specified in Config~\ref{config_llm} for \llm and Config~\ref{config_gim} for \gim, while varying $n$ and $t$, we obtain the following results by substituting the corresponding parameters into Theorem~\ref{thm:partial_recovery} and Theorem~\ref{thm:error_recover}:
\begin{itemize}
    \item For $(t,n)\in \{(9,2^{32}), (11,2^{26}), (13,2^{23}), (15,2^{20})\}$, we have $\pr{T}_{\mypartial} \le 2^{128}$.
    
    \item When $n=2^{24}$, we have $\pr{T}_{\ISD} \le 2^{122.21}$ for all choices of $t$.
\end{itemize}
\end{thm}

For Attack-I, it is natural to increase $n$ and $t$ to achieve a higher security level.
In Theorem~\ref{thm:mitigation_complexity}, we provide the minimum code length $n$ targeting at a security level of $\lambda= 128$ bits under different choices of $t$. 
However, it is noteworthy that increasing $t$ also reduces the error-correcting capability of the \ldpc. 
Therefore, a careful trade-off between security and robustness is required to effectively mitigate Attack-I.

For {Attack-III}, the most effective defense is to increase $n$ rather than $t$. 
In Theorem~\ref{thm:mitigation_complexity}, a code length $n>2^{24}$ is necessary to achieve $\lambda \ge 128$.
However, current state-of-the-art LLMs, such as GPT-4.1-Turbo~\cite{gpt41url}, have a maximum token length of $2^{15} = 32{,}768$, which corresponds to $n = 2^{20}$. 
The situation is similar for \gim. 
As discussed in~\cite{Rombach_2022_CVPR}, the maximum initial latent dimension is $4 \times 128 \times 128$, corresponding to a code length of only $n = 2^{16}$, still a significant gap from the target of $128$-bit security.
Consequently, constructing \ldpc-based watermarking schemes for LLMs and GIMs that achieve 128-bit security remains a significant practical challenge.

If these parameter requirements are satisfied, the resulting security can be reduced to the hardness of the underlying LPN and planted XOR problems, thereby providing protection not only against our attacks but also against other potential attacks targeting these primitives.

\subsection{Revision of Key Generation Algorithm}
\label{subsec:defense3}

The key strategy for mitigating Attack-II is to modify the $\KeyGen()$ function to reduce the probability of generating weak keys. 
We present a revised $\KeyGen()$ algorithm in~\Cref{alg:defense_watermark_keygeneration}, with modifications highlighted in \textcolor{applegreen}{\textbf{green}}. 
The revised algorithm first samples the secret key randomly and then constructs the public key $\mat{G}$ via Gaussian elimination. 
This process preserves the required randomness assumptions for both $\mat{P}$ and $\mat{G}$ in the \ldpc framework, making the probability of generating weak keys negligible.

It should be noted that this mitigation introduces a computational trade-off, i.e., the key generation complexity increases from $O(n^2)$ in the original algorithm to $O(n^3)$, owing to the Gaussian elimination step in Line~7.

\begin{algorithm}[t]
\caption{\textcolor{applegreen}{\textbf{Revised}} Key Generation Algorithm}
\label{alg:defense_watermark_keygeneration}
\begin{algorithmic}[1]

\Require 
Parameters: $n$, $r=0.99n$, $g$, $t$
\Ensure Matrices $\mat{P}$ and $\mat{G}$ for watermark generation
\State Initialize an empty list for rows of $\mat{P}$

\ForAll{$i \in [0.99n]$}
    \State Sample a random \textbf{\textcolor{applegreen}{$t$-sparse}} vector $\vec{s}_i \in \mathbb{F}_2^{n}$
    \State Append $\vec{s}_i$ to the list of rows for $\mat{P}$
\EndFor

\State Let $\mat{P}$ be the matrix whose rows are $\vec{s}_1, \ldots, \vec{s}_{0.99n}$
\State \textbf{\textcolor{applegreen}{Perform Gaussian elimination on $\mat{P}$ and transform it to system form $[\mat{I}_r|| \mat{P}_1]$}}
\State \textbf{\textcolor{applegreen}{Let $\mat{G}_0 = [\mat{P}_1^T || \mat{I}_{n-r}]^T$}}
\State \textbf{\textcolor{applegreen}{Random select $g$ columns from $\mat{G}_0$ to construct $\mat{G} \in \GF{2}^{n \times g}$}}
\State Sample a random permutation ${\Pi} \in \GF{2}^{n \times n}$, and let $\mat{P} \leftarrow \mat{P}\Pi^{-1}, \mat{G} \leftarrow \Pi\mat{ G} $
\State \Return $(\mat{P}, \mat{G})$

\end{algorithmic}
\end{algorithm}


\subsection{Implementation Suggestions for GIMs}
\label{subsec:defense2}
As discussed in~\cite{christ2024pseudorandom}, the condition $t=\Omega(\log n)$ is necessary to guarantee the security of PRC-based constructions, as it ensures that the entropy of the secret key space remains sufficiently large to resist exhaustive search attacks.
However, \gim~\cite{DBLP:conf/iclr/GunnZS25} introduces additional parity bits and employs a belief propagation decoding algorithm to recover the embedded message, which in turn enforces the parameter $t$ to be a constant in order to achieve reliable decoding.
When $t$ is reduced from $\Omega(\log n)$ to a constant $c$, the row-wise entropy of the secret key space $\binom{r}{t}$ significantly degrades, shrinking to $\Omega(r^c)=\Omega(n^c)$.
As a consequence, an adversary can mount a brute-force attack by enumerating all possible rows of the secret key with polynomial-time complexity ${O}(n^c)$.
Therefore, we recommend following the design principle in~\cite{christ2024pseudorandom} when constructing the multi-bit PRC and avoiding explicit decoding procedures such as $\mathsf{Decode}'()$.

Moreover, as discussed in~\Cref{subsec:concrete_complexity_attack1}, \gim adds a noise vector with weight $\eta \ll \rho$, which results in an extremely efficient concrete attack compared to the theoretical complexity $\pr{T}_{\ISD}$.
In fact, since $\lambda = g = \Theta(\log^2n)$ and $k$ could be seen as the sum of $g$ and a constant $c$ in \gim, the complexity of Attack-III, $\pr{T}'_{\ISD} = 2^{k/g} n^3 =2^{1+c/\Theta(\log^2n)} n^3 = O(n^3)$, is only \textit{polynomial}.
Therefore, we suggest increasing the weight of the added noise vector.


\section{Conclusion \& Future Work}
\label{sect/conclusion}
This paper presents the first cryptanalysis of PRC by proposing three attacks, including two distinguishing attacks and one noise overlay attack. 
Through rigorous theoretical analysis of success rates and complexity, we demonstrate that the current instantiation \ldpc exhibits security vulnerabilities. 
Our attack complexity is lower than its claimed security guarantees, with most attacks requiring fewer than $128$ bits of complexity. 
Furthermore, we implement PRC-based watermarking schemes in real-world large generative models to validate the practical threat. 
We also propose three defenses, including parameter and implementation recommendations, along with a revised key generation algorithm. 
Beyond security analysis, we refine PRC's parameter ranges and prove that even without attacks, the current PRC scheme faces challenges in practical LLM applications. 
Our findings will facilitate the development of future AIGC watermarking schemes that achieve robustness, undetectability, and practicality.

\mypara{Future Work}
Several open problems remain for future work.
First, as discussed in~\Cref{subsec:realworld_llm}, \llm is not directly applicable to real-world scenarios. Since our work focuses on security analysis, we have not yet provided concrete measures to enable the practical deployment of PRC in large-scale LLMs.
Second, the theoretical asymptotic analysis of PRC fails to provide security under practical feasible parameters. 
It might be promising to construct \gim focusing only on robustness instead of undetectability. 
Finally, regarding parameters $(n,g,t,r)$, because $r$ is typically set slightly below $n$ with negligible impact on security or error correction, we focus on $n$ and $t$ in our analysis of practical schemes. While we show that increasing $g$ to $k$ in \gim yields no asymptotic advantage, the precise discussion of $g$ remains open for further investigation.

\section*{Acknowledgement}
We sincerely thank the anonymous reviewers for their constructive suggestions. 
This work was supported by the National Cyber Security-National Science and Technology Major Project (2026ZD1500900), the Scientific Research Innovation Capability Support Project for Young Faculty (ZYGXQNJSKYCXNLZCXM-P4), the Fundamental and Interdisciplinary Disciplines Breakthrough Plan of the Ministry of Education of China (JYB2025XDXM114), the National Natural Science Foundation of China (62402273), and Zhongguancun Laboratory.

\appendix
\section*{Ethical Considerations}
Our ethical considerations consist of \textit{stakeholder analysis}, \textit{impact analysis}, \textit{mitigations}, and \textit{justifications}.

\mypara{Stakeholder Analysis}
Our cryptanalysis of PRCs involves three primary stakeholder groups as follows.
(1) \textit{Model Providers}: They provide large models, including LLMs and GIMs, to users.
Upon receiving requests from users, these models generate watermarked text or images, where the watermark is only detectable by model providers holding the secret key.
(2) \textit{End Users}: They utilize large models to generate content that contains watermarks which embed user-related information, such as user identifiers and timestamps, which may raise privacy concerns.
(3) \textit{Researchers}: This group includes researchers from both the security and AI communities.
They study whether watermarking mechanisms affect the quality of AI-generated content and whether the embedded watermarks may leak sensitive user information.

\mypara{Impact Analysis}
Our analysis has both positive and negative impacts for the three stakeholder groups identified.

\textit{Positive Impacts}:
(1) \textit{Enhanced Deployment Security} ( for model providers):  
Our analysis clarifies the security boundaries of PRCs under concrete parameter choices, helping model providers deploy PRC-based watermarking more carefully and avoid insecure configurations.
(2) \textit{Guided Research Directions} (for researchers):  
Our findings may inspire future work on strengthening \ldpc or on designing new PRC schemes based on alternative hardness assumptions, thereby advancing the robustness and reliability of watermarking research.

\textit{Negative Impacts}:
(1) \textit{Potential Privacy Leakage} (for end users): 
The attacks described in this work could, in principle, be misused to detect PRC watermarks, potentially leading to privacy leakage.  
Currently, \ldpc has not been widely deployed in LLMs due to limited error-correcting capability, so the immediate risk remains low.  
For diffusion-based generative models, improper parameter choices might allow adversaries to recover the secret key, compromising intended security guarantees.
(2) \textit{Perceived Confidence Reduction} (for researchers):  
Researchers may temporarily perceive \ldpc as less secure after our analysis.

\mypara{Mitigations}
To mitigate these potential risks, providers of large generative models should carefully evaluate the security of \ldpc under concrete parameter settings and apply the mitigations outlined in this work before any deployment. 
Responsible use entails rigorous parameter selection, comprehensive security analysis, and strict compliance with relevant legal and regulatory requirements.

\mypara{Justifications}
We argue that the benefits of this research outweigh the potential risks. 
PRC has been proposed as a foundational watermarking primitive, and it is therefore essential to conduct an independent security analysis before any widespread deployment. 
Our study demonstrates that the security of \ldpc has been overestimated across all examined parameter regimes. 
In particular, embedding PRC watermarks into LLMs is impractical, and implementations on GIMs suffer from inappropriate parameter choices and incorrect key generation procedures. 
Our work highlights these issues to prevent incorrect deployments while preserving the security of the ecosystem.

\section*{Open Science}
We release all source code and part of the experimental data at \url{https://github.com/1234wangtr/PRC_estimator} and permanently archived at \url{https://zenodo.org/records/20321793}.
The accompanying \texttt{README.md} provides detailed instructions on estimating the complexity of our attacks under concrete parameters in~\Cref{sect/concrete_complexity}, constructing the watermark schemes, and reproducing the real-world attacks described in~\Cref{sect/real_world_attack}.


\bibliographystyle{plainurl}
\bibliography{ref2}

\appendix







\section{Additional Details for PRC, LLM, and GIM}
\label{supp/llm_gim_watermark}

\subsection{Mathematical Notations}
The key mathematical notations are summarized in \Cref{tab:notations}.

\begin{table}[t]
\centering
\caption{Key Mathematical Notations.}
\label{tab:notations}
\setlength{\tabcolsep}{2pt}
\renewcommand{\arraystretch}{1.2}
\begin{tabular}{c|p{6cm}}
\toprule
\textbf{Notation} & \textbf{Description} \\
\midrule
$\lambda$ & Security parameter \\
$n$ & Code length of PRC \\
$g$ & The column size of matrix $\mat{G}$ of PRC \\
$k$ & The column size of matrix $\mat{G}$ of \gim \\
$t$ & Row-wise nonzero count of $\mat{P}$ \\
$r$ & The row size of secret key $\mat{P}$ \\
$\mat{G}$ & $\mat{G}\in \mathbb{F}^{n\times g}_2$ is a public generator matrix \\
$\mat{P}$ & $\mat{P} \in \mathbb{F}^{r\times n}_2$ is a secret parity-check matrix \\
$\vec{z}$ & $\mat{z} \in \GF{2}^{n}$ is a random vector serving as the one-time pad \\
$\mathsf{pk}$ & $\mathsf{pk}=(\mat{G},\mat{z})$ is the public key of PRC\\
$\mathsf{sk}$ & $\mathsf{sk}=(\mat{P},\mat{z})$ is the secret key of PRC\\
$\wt{\vec{y}}$ & Hamming weight of a vector $\vec{y}$ \\
$\Ber(p)$ & Bernoulli distribution on $\{0,1\}$ with expectation $0 \le p \le 1$ \\ 
$\Ber(n,p)$ & $n$-bit strings where each bit is an i.i.d sample from $\Ber(p)$ \\
$\langle\vec{x},\vec{y}\rangle$ & Component-wise product of $\vec{x}$ and $\vec{y}$ \\
\bottomrule
\end{tabular}
\end{table}

\subsection{Detailed Algorithms for LLM and GIM}
The detailed algorithms for LLM and GIM generation are provided in Algorithm~\ref{alg:LLM} and Algorithm~\ref{alg:GIM}, respectively.

\begin{algorithm}[]
\caption{\texttt{LLM Generation}}
\label{alg:LLM}
\begin{algorithmic}[1]
\Require $\mathsf{prompt}\in \mathcal{T}^*$
\Param Temperature $\mathsf{temp}\in \mathbb R^+$, generation length $\mathsf{len}\in \mathbb Z^+$
\Ensure Generated text $T\in \mathcal{T}^*$
\State $T \leftarrow \perp$
\For{$i\leftarrow 0$ to $\mathsf{len}-1$}
\State ${\vec l}\leftarrow \mathsf{LLM}(\mathsf{prompt}\parallel T)$
\Comment The logits of tokens
\State ${\vec p}\leftarrow \mathsf{Softmax}(\vec l / \mathsf{temp})$ \Comment The probabilities of tokens 
\State $t\overset{\vec p}{\leftarrow}\mathcal{T}$ \Comment Sample the next token
\State $T\leftarrow T\parallel t$
\EndFor
\State \Return $T$
\end{algorithmic}
\end{algorithm}
\begin{algorithm}[t]
\caption{\texttt{GIM Generation}}
\label{alg:GIM}
\begin{algorithmic}[1]
\Require $\mathsf{prompt}\in \mathcal{T}^*$
\Param De-noising model \texttt{U-net}, Autoencoder model (\texttt{E-net},~\texttt{D-net}), 
\Statex  de-noising steps $s$, a (possibly randomized) de-noising scheduler $f(\cdot)$
\Ensure Generated image $\mathsf{Img}$
\State $\vec y^{s} \overset{\$}{\leftarrow}\mathcal{N}(\vec 0,{\bf I}_n)$
\For{$i \leftarrow s$ down to $1$}
\State $\vec y^{i-1} \leftarrow f(\texttt{U-net}(\mathsf{prompt},\vec y^{i}, i),i)$ 
\EndFor
\State $\mathsf{Img}\leftarrow \texttt{D-net}(\vec y^0)$
\State \Return $\mathsf{Img}$
\end{algorithmic}
\end{algorithm}

\subsection{Detailed Algorithms for \llm and \gim}
The specific algorithms involved in \llm and \gim are further described in Algorithm~\ref{alg:TokenSample}, Algorithm~\ref{alg:TextInv}, and Algorithm~\ref{alg:GIMInv}.
\begin{algorithm}[]
\caption{\texttt{Token Sampling}}
\label{alg:TokenSample}
\begin{algorithmic}[1]
\Require Token probabilities $\vec p$
\Statex codeword trunk $\vec x'$
\Param token space $\mathcal{T}$
\Ensure Sampled token $t\in \mathcal{T}$
\State $b \leftarrow \perp$
\For{$i\leftarrow 0$ to $\lceil \log_2|\mathcal{T}|\rceil-1$}
\State $p'\leftarrow \frac{\sum_{t\in \mathcal{T}\land \mathsf{Unpack}(t)=b\parallel1\parallel *}p_t}{\sum_{t\in \mathcal{T}\land \mathsf{Unpack}(t)=b\parallel *}p_t}$
\Statex \Comment{Marginal probability of current bit}
\If{$p'\le 1/2$}
\State $t'\leftarrow \Ber(2p'x_i')$
\Else
\State $t'\leftarrow \Ber(1-2(1-p')(1-x'_i))$
\EndIf
\State $b\leftarrow b\parallel t'$
\Comment{Sample token bit-wise}
\EndFor
\State $t\leftarrow \mathsf{Packbits}(b)$
\State \Return $t$
\end{algorithmic}
\end{algorithm}

\begin{algorithm}[]
\caption{\texttt{TextToCodeword}}
\label{alg:TextInv}
\begin{algorithmic}[1]
\Require Text $T\in\mathcal{T}^\mathsf{len}$
\Ensure Recovered codeword $\vec x'\in \mathbb{R}^n$
\State $\vec x' = \perp$
\For{$i\leftarrow 0$ to $\mathsf{len}-1$}
\State $\vec x'\leftarrow \vec x' \parallel \mathsf{Unpackbits}(T_i)$
\EndFor
\State \Return $\vec x'$
\end{algorithmic}
\end{algorithm}

\begin{algorithm}[]
\caption{\texttt{ImageToCodeword}}
\label{alg:GIMInv}
\begin{algorithmic}[1]
\Require Image $\mathsf{Img}$
\Param De-noising model \texttt{U-net}, Autoencoder model (\texttt{E-net},~\texttt{D-net}), 
\Statex inversion steps $s$, a (possibly randomized) noise inversion scheduler $z(\cdot)$, error calibration factor $\sigma$
\Ensure Recovered codeword $\vec x'\in \mathbb{R}^n$
\State $\vec y^{0}\leftarrow \texttt{E-net}(\mathsf{Img})$
\For{$i \leftarrow 0$ to $s-1$}
\State $\vec y^{i+1} \leftarrow z(\texttt{U-net}(\vec y^{i}, \perp,i+1),i+1)$ 
\EndFor
\State $\vec x'=\mathsf{erf}(\vec y^{s}/\sqrt{2\sigma^2(1+\sigma^2)})$
\State \Return $\vec x'$
\end{algorithmic}
\end{algorithm}

\section{Procedures of Attacks}
\label{supp/attack_algorithm}
This section describes the full procedures of the three attacks, which correspond to Attack-I, Attack-II, and Attack-III, respectively, as shown in Algorithm~\ref{alg:attack2_partial}, Algorithm~\ref{alg:attack3_weakey} and Algorithm~\ref{alg:attack4_erroroverlay}.

\begin{algorithm}[]
\caption{\texttt{Attack-I} }
\label{alg:attack2_partial}
\begin{algorithmic}[1]
\Require Public Key $(\mat{G},\vec{z})$, target codewords $\myset{X} = \{\vec{x}^{(1)},\cdots,\vec{x}^{(m)}\}$ 
\Param $(\alpha,\beta,r',l, \tau)$ defined in~\Cref{sec:partial_sk_recovery} and $(g,t,n,r)$ of \ldpc
\Ensure $\mathsf{true}$ or $\mathsf{false}$, where $\mathsf{true}$ indicates $\myset{X}$ are encoded by PRC and $\mathsf{false}$ otherwise
\State $\myset{L}_1 \leftarrow  \{ (a_1,...,a_{t_1},\vec{sum}_1) | 1\leq a_1 < ... < a_{t_1} \leq n, \vec{sum}_1=\sum_{i=1}^{t_1} \mat{G}_{a_i} \}$
\State $\myset{L}_2 \leftarrow \{ (a_{t_1+1},...,a_{t},\vec{sum}_2) | 1 \leq a_{t_1+1} < ... < a_{t} \leq n,  \vec{sum}_2=\sum_{i=t_1+1}^{t} \mat{G}_{a_i} \}$
\If{$t$ is \textbf{even}}
\State Construct $\myset{L}_1' \subset \myset{L}_1$ and $|\myset{L}_1'|=\frac{1}{\sqrt{r'/l}}|\myset{L}_1|$
\State Construct $\myset{L}_2' \subset \myset{L}_2$ and $|\myset{L}_2'|=\frac{1}{\sqrt{r'/l}}|\myset{L}_2|$
\Else{}
\State Construct $\myset{L}_1' \subset \myset{L}_1$ and $|\myset{L}_1'|={\alpha}|\myset{L}_1|$
\State Construct $\myset{L}_2' \subset \myset{L}_2$ and $|\myset{L}_2'|={\beta}|\myset{L}_2|$
\EndIf
\State Sort $\myset{L}'_1,\myset{L}'_2$
\State $\myset{L}' \leftarrow \text{MERGE-JOIN}(\myset{L}'_1,\myset{L}'_2)$
\State Initialize an empty list for vectors of $\myset{V}$
\ForAll{$(a_1,\cdots,a_t) \in \myset{L}'$}
\State $\vec{v} \leftarrow (v_1,\cdots,v_n)$ 
\Comment{where $v_i=1$ if $\exists j \text{ s.t. } i=a_j$ and $v_i=0$ otherwise}
\State Append $\vec{v}$ to the list of vectors for $\myset{V}$
\EndFor
\State Initialize $N_\mathsf{tot},N_\mathsf{zero} \leftarrow 0,0$
\ForAll{$\vec{v}^{(j)} \in \myset{V}$ and $\vec{x}^{(i)} \in \myset{X}$ }
\State $N_{\mathsf{tot}} \leftarrow N_{\mathsf{tot}} + 1$
\If{ $\langle \vec{v}^{(j)}, \vec{x}^{(i)}+\vec{z}  \rangle = 0 $  }
\State $N_{\mathsf{zero}} \leftarrow N_{\mathsf{zero}} + 1$
\EndIf
\EndFor
\If {$N_{\mathsf{zero}} / N_{\mathsf{tot}} \ge \tau $} 
\State \Return $\mathsf{true}$  
\Else {} 
\State \Return $\mathsf{false}$ 
\EndIf
\end{algorithmic}
\end{algorithm}

\begin{algorithm}[]
\caption{\texttt{Attack-II}}
\label{alg:attack3_weakey}
\begin{algorithmic}[1]
\Require Public Keys $\{(\mat{G}^{(i)},\vec{z}^{(i)})\}$, target codewords $\myset{X} = \{\vec{x}^{(1)},\cdots,\vec{x}^{(m)}\}$ outputted by a certain oracle 
\Param Threshold $\tau$ defined in Section~\ref{sec:partial_sk_recovery} and $(g,t,n,r)$ of \ldpc
\Ensure $\mathsf{true}$ or $\mathsf{false}$, where $\mathsf{true}$ indicates $\myset{X}$ are generated by \ldpc and $\mathsf{false}$ otherwise
\State Randomly select a public key pair $(\mat{G},\vec{z}) \xleftarrow{\$} \{(\mat{G}^{(i)},\vec{z}^{(i)})\} $
\State $\myset{I} \leftarrow \{ (\alpha_j,\beta_j) | \mat{G}_{\alpha_j} = \mat{G}_{\beta_j} \}$
\If {$|\myset{I}|=0$}
\State goto line 1
\EndIf
\State Obtain $\myset{X}$ from the oracle corresponding to $(\mat{G},\vec{z})$
\State Initialize $N_\mathsf{tot},N_\mathsf{dup} \leftarrow 0,0$
\ForAll{$(\alpha_j,\beta_j) \in \myset{I}$ and $\vec{x}^{(i)} \in \myset{X}$ }
\State $N_{\mathsf{tot}} \leftarrow N_{\mathsf{tot}} + 1$
\If{ $x^{(i)}_{\alpha_j}= x^{(i)}_{\beta_j} $  }
\State $N_{\mathsf{dup}} \leftarrow N_{\mathsf{dup}} + 1$
\EndIf
\EndFor
\If {$N_{\mathsf{dup}} / N_{\mathsf{tot}} \ge \tau $} 
\State \Return $\mathsf{true}$  
\Else {} 
\State \Return $\mathsf{false}$ 
\EndIf

\end{algorithmic}
\end{algorithm}

\begin{algorithm}[]
\caption{\texttt{Attack-III}}
\label{alg:attack4_erroroverlay}
\begin{algorithmic}[1]
\Require Public Key $(\mat{G},\vec{z})$, target codeword $\vec{x}=\mat{G}\vec{s}+\vec{e}+\vec{z}$
\Param noise rate $\mu$ defined in~\Cref{subsec:msg_recovery} and $(g,t,n,r)$ of \ldpc
\Ensure $\vec{x}''=\vec{x}+\vec{e}'$ s.t. $\Decode(1^\lambda,\mathsf{sk},\vec{x}'') = \perp$ and $w_H(\vec{e}')=\mu n$
\State $\vec{x}' \leftarrow \vec{x}+\vec{z}$
\State $\vec{e} \leftarrow \mathsf{ISD}(\mat{G},\vec{x}')$ 
\Comment{recover $\vec{e}$ from $\vec{x}',\mat{G}$ with an ISD algorithm}
\State $\myset{S} \leftarrow \{1,\cdots,n\} \backslash \mathsf{supp}(\vec{e})$ 
\Comment{where $\mathsf{supp}(\vec{e})$ denotes the support set of $\vec{e}$ }
\State Choose $\myset{T} \subset \myset{S}$ and $|\myset{T}|=\mu n$
\State Construct $\vec{e}'$ with support set $\myset{T}$
\State $\vec{x}'' \leftarrow \vec{x}+\vec{e}'$
\State \Return $\vec{x}''$
\end{algorithmic}
\end{algorithm}

\section{Complexities of Attacks}
\label{supp/attack_complexity}
The detailed complexities of our attacks are summarized in Table~\ref{tab:complexity_llms} for \llm and Table~\ref{tab:complexity_gims} for \gim.

\begin{table*}[t]
\centering
\caption{Time complexities of all attacks on \llm.}
\setlength{\tabcolsep}{13pt}
\begin{tabular}{c|c|c|c|c|c|c|c}
\toprule
$t$ & $\epsilon$ & $\rho$ & $\log_2^{\pr{T}_{\mypartial}}$ & $\log_2^{\pr{P}_{\weak}^{\rm{LLM}}}$ & $\log_2^{\pr{T}_{\dis}}$ & $\log_2^{\pr{T}_{\ISD}}$ & $\lambda$ \\  
\midrule \midrule
$3$ & $0.236$ & $0.264$ & $21.30$ & $0.00$ & $22.58$ & $72.21$ & $48$\\
$4$ & $0.285$ & $0.215$ & $29.19$ & $-0.00$ & $22.98$ & $73.02$ & $63$\\
$5$ & $0.319$ & $0.181$ & $36.84$ & $-3.91$ & $27.20$ & $73.50$ & $78$\\
$6$ & $0.344$ & $0.156$ & $44.28$ & $-11.89$ & $35.42$ & $73.57$ & $93$\\
$7$ & $0.363$ & $0.137$ & $51.59$ & $-19.66$ & $43.39$ & $73.61$ & $107$\\
$8$ & $0.377$ & $0.123$ & $58.76$ & $-27.20$ & $51.11$ & $73.64$ & $121$\\
$9$ & $0.389$ & $0.111$ & $65.84$ & $-34.55$ & $58.62$ & $73.66$ & $135$\\
$10$ & $0.399$ & $0.101$ & $72.82$ & $-41.73$ & $65.94$ & $73.67$ & $148$\\
$11$ & $0.408$ & $0.092$ & $79.71$ & $-48.75$ & $73.08$ & $73.54$ & $162$\\
$12$ & $0.415$ & $0.085$ & $86.54$ & $-55.64$ & $80.09$ & $73.56$ & $175$\\
$13$ & $0.421$ & $0.079$ & $93.29$ & $-62.40$ & $86.95$ & $73.46$ & $188$\\
$14$ & $0.426$ & $0.074$ & $99.99$ & $-69.04$ & $93.69$ & $73.37$ & $202$\\
\bottomrule
\end{tabular}
\label{tab:complexity_llms}
\end{table*}

\begin{table*}[]
\centering
\caption{Time complexities of all attacks on \gim.}
\setlength{\tabcolsep}{10pt}
\begin{tabular}{c|c|c|c|c|c|c|c|c|c}
\toprule
$t$ & $\epsilon$ & $\rho$ & $\eta$ & $\log_2^{\pr{T}_{\mypartial}}$ & $\log_2^{\pr{P}_{\weak}^{\rm{GIM}}}$ & $\log_2^{\pr{T}_{\dis}}$ & $\log_2^{\pr{T}_{\ISD}}$ & $\log_2^{\pr{T}_{\ISD}'}$ & $\lambda$ \\  
\midrule \midrule
$3$ & $0.281$ & $0.219$ & $0.018$ & $21.88$ & $-0.01$ & $23.16$ & $244.03$ & $56.56$ & $39$\\
$4$ & $0.325$ & $0.175$ & $0.013$ & $27.92$ & $-7.83$ & $31.01$ & $202.96$ & $53.37$ & $51$\\
$5$ & $0.354$ & $0.146$ & $0.011$ & $33.79$ & $-16.71$ & $39.92$ & $176.52$ & $51.40$ & $63$\\
$6$ & $0.375$ & $0.125$ & $0.009$ & $39.52$ & $-24.87$ & $48.11$ & $157.93$ & $50.15$ & $75$\\
$7$ & $0.391$ & $0.109$ & $0.008$ & $45.14$ & $-32.61$ & $55.87$ & $144.27$ & $49.22$ & $86$\\
\bottomrule
\end{tabular}
\label{tab:complexity_gims}
\end{table*}

\section{Attack Performance Under Proxy Models}
\label{supp/diff_ddim_inv}
When performing our attacks under proxy models, we generate watermarked images with stable-diffusion-2-1-base, while extracting the codewords through~\Cref{alg:GIMInv} using the autoencoders (\texttt{E-net},\texttt{D-net}) from stable-diffusion-2-1-base ($\mathsf{SD21}$), stable-diffusion-v1-5 ($\mathsf{SD15}$), and stable-diffusion-2-base ($\mathsf{SD20}$)~\cite{Rombach_2022_CVPR}.
Then, we calculate TPR values of $\text{TPR}_{\mathsf{SD21}},\text{TPR}_{\mathsf{SD15}}$, and $\text{TPR}_{\mathsf{SD20}}$. 
The parameter choice is the same as the formal set except that $\tau=0.60$ for both Attack-I and Attack-II.

The results are presented in Table~\ref{tab:diff_ddim_new}.
We display only one FPR column, as the FPR is $0\%$ for all vectors in the experiments.
Our attack with proxy neural networks remains effective in both Attack-I and Attack-II.
In particular, we observe that $\mathsf{SD20}$ performs almost as well as the original target model $\mathsf{SD21}$.
These experimental results further demonstrate the practical threats posed by our proposed watermark removal attacks.

\begin{table}[p]
\centering
\caption{Watermark detection attacks against real-world GIM with different proxy models when $t=3$.
The success rates of all attacks are $100\%$.
}
\label{tab:diff_ddim_new}

\setlength{\tabcolsep}{5pt}
\begin{tabular}{c|c|c|c|c|c}
\toprule
Attack & $\text{TPR}_0$ & $\text{TPR}_{\mathsf{SD21}}$ & $\text{TPR}_\mathsf{SD15}$ & $\text{TPR}_\mathsf{SD20}$ & FPR   \\ 
\midrule
 I  & $98\%$ & $95\%$ & $61\%$ & $95\%$ & $0\%$  \\ 
 II & $100\%$ & $100\%$ & $95\%$ & $100\%$ & $0\%$  \\ 
\bottomrule
\end{tabular}
\end{table}

\section{Attack Performance Under Different Models}
\label{supp/diff_watermarks}

We conduct extra experiments using different generative models under the same setup as in~\Cref{subsec:setups} with $t=3$.
We employ Qwen-3B~\cite{qwen3technicalreport} for \llm, where Attack-I achieves $\text{TPR}_0=99\%,\text{TPR}_1=33\%,\text{FPR}=1\%$ and Attack-II achieves $\text{TPR}_0=100\%,\text{TPR}_1=100\%,\text{FPR}=0\%$.
It is worth noting that \llm also exhibits impractical quality on Qwen-3B.
For \gim, we utilize stable-diffusion-v1-5~\cite{Rombach_2022_CVPR}, where $\text{TPR}_0=98\%,\text{TPR}_1=84\%,\text{FPR}=2\%$ and Attack-II achieves $\text{TPR}_0=100\%,\text{TPR}_1=97\%,\text{FPR}=0\%$.
In all cases, the success rate is $100\%$.
The gap between TPR and FPR demonstrates that our distinguishing attacks remain effective across diverse models.




\clearpage

\twocolumn[
\begin{center}
\section{Artifact Appendix}
\end{center}
]

\subsection{Abstract}
Pseudorandom error-correcting codes (PRCs), a novel cryptographic primitive recently proposed at CRYPTO 2024, are primarily applied in undetectable watermarking schemes for large generative models. However, the security of PRCs has not yet been systematically analyzed. To fill this gap, we present the first cryptanalysis of PRCs. Specifically, focusing on LDPC-PRC, the only known practical instantiation of PRCs, we propose three novel attacks that challenge its undetectability (Attack I, II) and robustness (Attack III).
To rigorously demonstrate the practical threat, this artifact analyzes the concrete attack complexity under realistic parameters and validates the attack effectiveness on both real-world large language models and generative image models, including DeepSeek and Stable Diffusion.
Our analysis shows that the claimed security guarantees of LDPC-PRC are undermined across all practically feasible regimes.

\subsection{Description \& Requirements}

\subsubsection{Security, privacy, and ethical concerns}
The artifact attacks the PRC watermark scheme deployed in self-hosted AIGC systems without involving any local data, real-world systems or human subjects.
Therefore, it does not raise security, privacy, or ethical concerns.

\subsubsection{How to access}
Our artifact is permanently and publicly available at \url{https://zenodo.org/records/21771706} with DOI:10.5281/zenodo.21771706.
Our artifact can also be downloaded from GitHub with:\\
{\small \texttt
{git clone https://github.com/1234wangtr/PRC\_estimator
}}

\subsubsection{Hardware dependencies}
Fully running this artifact requires a GPU machine with at least 80GB VRAM (Ours: NVIDIA H20 with 141GB of VRAM), equipped with a CPU with moderate computational power (Ours: Intel XEON PLATINUM 8558P 2.70 GHz), 16GB RAM (Ours: 2TB), and 60GB free disk space. However, the main claims can be verified on a CPU-only machine with 5 GB of free disk space, given pre-generated data.

\subsubsection{Software dependencies}
\begin{compactdesc}
  \item \textbf{OS:} Modern x86\_64 Linux with \texttt{git}, \texttt{bash}, and \texttt{unzip} installed (Ours: Ubuntu 24.04.1, Kernel 6.17.0-29-generic).
  \item \textbf{Python:} 3.11, with \texttt{pip} installed and supports \texttt{conda} virtual environments.  (Ours: 3.11.4 with Miniconda 25.3.1)
  \item \textbf{(Optional for GPU) NVIDIA Driver:} 525.60.13+ to support CUDA 12.x. (Ours: 595.71.05, CUDA 12.1.105)
  \item \textbf{(Optional) Fonts:} \textit{Times New Roman} and \textit{Droid Sans Fallback}. You can install it via \texttt{setup/get\_fonts.sh}.
\end{compactdesc}

\subsubsection{Benchmarks}
None. This artifact does not involve third-party benchmarks.

\subsection{Set-up}
\subsubsection{Installation}
\label{sec:install}
After downloading the artifact, you can enter \texttt{PRC\_estimator} folder, which will be the working directory for all the following commands.

Our artifact involves two separate Python virtual environments for LLM and GIM experiments. You can set them up with:
\texttt{
conda env create -f llm/environment.yml \&\&
conda env create -f gim/environment.yml}

\subsubsection{Basic Test}
You can run Experiment 1 (E1) to test if the environments are correctly set up.

\subsection{Evaluation workflow}
\subsubsection{Major Claims}
\begin{compactdesc}
  \item[(C1):] Our attacks reveal that PRC-based watermarking schemes failed to achieve 128-bit security under realistic parameters.
  This is proven by the experiment E1 described in Section~5.2 of our paper, whose results are presented in Figure~4 and Tables~6 and~7 of our paper.

  \item[(C2):] Our Attacks I and II can successfully distinguish real-world LDPC-PRC watermarks for generated text.
  This is proven by the experiment E2 described in Section~6.2.2 and Section~6.2.3 of our paper, whose results are presented in Table~3 of our paper. 

  \item[(C3):] Employing LDPC-PRC watermarks for text generation only works for large temperatures, which leads to a significant quality degradation.

  \item[(C4):] Our Attacks I and II can successfully distinguish real-world LDPC-PRC watermarks for generated images.
  This is proven by the experiment E4 described in Section~6.3.2 of our paper, whose results are presented in Table~4 of our paper.
  These attacks can success even when the adversary has access only to a proxy model.
  This is proven by the experiment E4 described in Appendix D of our paper, whose results are presented in Table~8 of our paper.

  \item[(C5):] The attacker can perform Attack III to remove the LDPC-PRC watermarks for images with a high success rate under a reasonable image quality budget.
  This is proven by the experiment E5 described in Section~6.3.2 of our paper, whose results are presented at the end of Section~6.3.2 of our paper.
\end{compactdesc}

\subsubsection{Experiments}

\begin{compactdesc}
  \item[(E1):] \textbf{Concrete Time Complexity Analysis} [1 human-minute + 1 CPU minute].
  Estimate the concrete time complexity of our attacks against the LDPC-PRC schemes with LLM and GIM watermarking parameters.
  
  \textbf{Execution:} Run the security estimation scripts as:
  \begin{verbatim}
conda run -n prc-estimator-llm \
  python llm/security_estim/main.py
conda run -n prc-estimator-gim \
  python gim/security_estim/main.py
\end{verbatim}

  \textbf{Results:} The figures of attack time complexity are in \texttt{<llm or gim>/data/security\_estim.pdf}, corresponding to Figure~4 (a) and (b).
  The detailed results are in \texttt{<llm or gim>/data/security\_estim.csv}, corresponding to Tables~6 and~7.

  \item[(E2):] \textbf{Attack I and II against LDPC-PRC Watermarked Text} [1 human-minute + 3 CPU-minutes/10 groups]. Perform Attacks I and II against the LDPC-PRC watermarked text generated by DeepSeek-R1-Distill-Qwen-7B under $t=3$ and temperature 1.8.

  \textbf{Using pre-generated victim data:} Since generating watermarked text requires significant GPU resources, we provide 10 groups of pre-generated watermarked text under $t=3$. You can obtain the data by:
  {\small 
  \begin{verbatim}
unzip llm/data/Deepseek_t_3_temp_1.8.zip \
  -d llm/data/Deepseek_t_3_temp_1.8_example
\end{verbatim}
}

  \textbf{(Optional) Self-generating victim data:} [Extra 5 GPU-hours and 0.2GB disk space for 10 groups, extra 16GB disk space for model] First, download the model using \texttt{setup/get\_llm.sh}. Then, generate watermarked text with:
{\small
\begin{verbatim}
conda activate prc-estimator-llm
export CUDA_VISIBLE_DEVICES=0
python llm/generation/main.py \
  --temperature 1.8  --model_name Deepseek \
  --prc_t 3 --start 0 --end 10
\end{verbatim}
}
This will generate 10 groups of watermarked text with different $\mathsf{pk}$ to \texttt{llm/data/Deepseek\_t\_3\_temp\_1.8}. 

  \textbf{Execution:} The attacks can be performed with:
{\small
\begin{verbatim}
conda activate prc-estimator-llm
python llm/attack/attack1_2_main.py --t 3 \
  llm/data/Deepseek_t_3_temp_1.8[_example]
\end{verbatim}
}
This attack aims to distinguish PRC codewords and random vectors.

\textbf{Results:} 
The attack results are in \texttt{llm/data/Deepseek\_} \texttt{t\_3\_temp\_1.8[\_example]/result.csv}, corresponding to Table~3.
The meaning of each field is as follows:
\begin{compactdesc}
    \item \textbf{$\text{TPR}_0$:} The true-positive rate at which PRC codewords are correctly classified (higher is better).
    \item \textbf{$\text{FPR}_0$:} The false-positive rate at which random vectors are incorrectly classified as PRC codewords (lower is better).
    \item \textbf{$\text{TPR}_1$:} Measures whether the inversed PRC codeword can be detected as a watermark. A significant gap between $\text{TPR}_1$ and $\text{FPR}_0$ demonstrates better effectiveness of the attack.
    \item \textbf{Success Rate:} For Attack-I, it corresponds to the proportion of instances in which at least one row of the secret key is successfully recovered. For Attack-II, it corresponds to the proportion of instances in which at least one pair of identical rows is successfully identified.
\end{compactdesc}

Although our paper reports results over 128 groups, the results over 10 groups should have the same trend.

The attack details are in \texttt{llm/data/Deepseek\_t\_3} \texttt{\_temp\_1.8[\_example]/log.txt}.
You can inspect some intermediate outputs such as:
{\small
  \begin{verbatim}
# Recovered nonzero positions of the secret key
grep "found l1_vec" \ 
llm/data/Deepseek_t_3_temp_1.8[_example]/log.txt
# Duplicated positions identified from pub key
grep "dup_dict" \
llm/data/Deepseek_t_3_temp_1.8[_example]/log.txt
\end{verbatim}
}

  \item[(E3):] \textbf{Watermarked Text Quality Analysis} [2 human-minutes + 2 CPU minutes]. Analyze the pre-generated quality of the generated watermarked text.
  
  \textbf{Preparation:} Since the quality analysis requires lots of watermarked text generated under different temperatures, the analysis is performed on the pre-generated data, which can be obtained by:
  {\small
  \begin{verbatim}
unzip -d llm/data \
  llm/data/Deepseek_t_3_temp_all.zip 
\end{verbatim}
}
  \textbf{Execution:} Run the following command to analyze the quality of the generated watermarked text.
  {\small
  \begin{verbatim}
conda run -n prc-estimator-llm \
  python llm/generation/plot_entropy.py
\end{verbatim}
}

  \textbf{Results:} The generated figure is in \texttt{llm/data/avg\_} \texttt{entropy\_vs\_correct\_rate.pdf}, corresponding to Figure~5.

  \item[(E4):] \textbf{Attack I and II against LDPC-PRC Watermarked Image} [1 human-minute + 6 CPU-minutes/10 groups]. Perform Attacks I and II against the LDPC-PRC watermarked image generated under $t=3$.

  \textbf{Using pre-generated victim data:} Since generating watermarked images requires significant GPU resources, we provide 10 groups of pre-generated watermarked images under $t=3$. You can obtain the data by:
  {\small
  \begin{verbatim}
unzip gim/data/SD21_t3.zip \
  -d gim/data/SD21_t3_example
unzip gim/data/SD21_t3_inv_SD15_SD2_SD21.zip \
  -d gim/data/SD21_t3_inv_SD15_SD2_SD21_example
\end{verbatim}
}

  \textbf{(Optional) Self-generating victim data:} [Extra 2 GPU-hour, 1GB disk space for 10 groups, extra 25GB disk space for model] First, download the model and dataset with \texttt{setup/get\_gim.sh}. Then, generate watermarked images with:
{\small
\begin{verbatim}
conda activate prc-estimator-gim
export CUDA_VISIBLE_DEVICES=0
python gim/generation/main.py \
  --prc_t 3 --start 0 --end 10 \
  --gen_model_id SD21 --inv_model_ids SD21
# using proxy models
python gim/generation/main.py \
  --start 0 --end 10 --gen_model_id SD21 \
  --inv_model_ids SD15,SD2,SD21 --prc_t 3 
\end{verbatim}
}
This will generate 10 groups of watermarked images with different $\mathsf{pk}$ to \texttt{gim/data/SD21\_t3}, along with 10 groups of watermarked images inverted by 3 different proxy models to \texttt{gim/data/} \texttt{SD21\_t3\_inv\_SD15\_SD2\_SD21}

  \textbf{Execution:} The attacks can be performed with:

{\small
\begin{verbatim}
conda activate prc-estimator-gim
python gim/attack/attack1_2_main.py --t 3 \
  gim/data/SD21_t3[_example]
python gim/attack/attack1_2_diff_inv_main.py \
  gim/data/SD21_t3_inv_SD15_SD2_SD21[_example] \
  --t 3 # using proxy models
\end{verbatim}
}

This attack aims to distinguish PRC codewords and random vectors.

\textbf{Results:} 
The attack results without using proxy model are in \texttt{gim/data/SD21\_t3[\_example]/result.csv}, corresponding to Table~4.
The attack results with proxy model are in \texttt{gim/data/SD21\_t3\_inv\_SD15\_SD2\_} \texttt{SD21[\_example]/result.csv}, corresponding to Table~8.
The meaning of these results are the same as described in  (E2). Although our paper reports results over 128 groups, the results over 10 groups should have the same trend.
You can inspect some intermediate outputs at the corresponding \texttt{log.txt} files as demonstrated in (E2).

  \item[(E5):] \textbf{Watermark Removal Analysis} [5 human-minutes +
  2 GPU-hour for 10 images].
  
  \textbf{Preparation:} Make sure the model and dataset are downloaded with \texttt{setup/get\_gim.sh}.

  \textbf{Execution:} First, generate watermarked images with:
{\small
\begin{verbatim}
conda activate prc-estimator-gim
export CUDA_VISIBLE_DEVICES=0
python gim/generation/main-for-attack3.py \
  --prc_t 3 --start 0 --end 10 \
  --model_id SD21
\end{verbatim}
}
This will generate 10 watermarked images and save intermediate results to \texttt{gim/data/attack3\_SD21\_t3} in 2 GPU-minutes.

Then, generate the adversarial images with $\epsilon = 16.0$ based on those intermediate results (corresponding to the attack results for Attack III) with:
{\small
\begin{verbatim}
conda activate prc-estimator-gim
export CUDA_VISIBLE_DEVICES=0
python gim/attack/attack3_main.py \
  gim/data/attack3_SD21_t3 --start 0 --end 10 \
  --model_id SD21 --eps 16.0
\end{verbatim}
}

The attack will run for 1 GPU-hour for 10 images.

\textbf{Results:}
The generated adversarial images are in \texttt{gim/data/attack3\_SD21\_t3/inv\_lat\_16.0}, and the images with removed watermarks are in \texttt{gim/data/attack3\_SD21\_t3/adv\_img\_16.0}.

To calculate the removal success rate, run:
{\small
\begin{verbatim}
conda activate prc-estimator-gim
python gim/attack/attack3_stati.py \
  gim/data/attack3_SD21_t3/inv_lat_16.0
\end{verbatim}
}

This result corresponds to the attack success rate reported at the end of Section~6.3.2.
\end{compactdesc}

\subsection{Version}
Based on the LaTeX template for Artifact Evaluation V20231005. Submission,
reviewing and badging methodology followed for the evaluation of this artifact
can be found at \url{https://secartifacts.github.io/usenixsec2026/}.




\end{document}